\documentclass[11pt]{article}
\usepackage{times}
\usepackage{paralist}
\usepackage{amsfonts}
\usepackage{amsmath}
\usepackage{amsthm}
\usepackage{latexsym}
\usepackage{url,hyperref}
\usepackage{picins,boxedminipage}

\advance \topmargin by -0.7in

\newtheorem{theorem}{Theorem}[section]
\newtheorem{lemma}[theorem]{Lemma}
\newtheorem{proposition}[theorem]{Proposition}
\newtheorem{corollary}[theorem]{Corollary}

\newcommand{\partdiff}[2]{\frac{\partial {#1}}{\partial {#2}}}

\newcommand{\mixdiff}[3]{\frac{\partial^2 {#1}}{{\partial {#2}}{\partial {#3}}}}

\newcommand{\eps}{\varepsilon}
\newcommand{\E}{\mbox{\bf E}}

\newcommand{\cI}{{\cal I}}
\newcommand{\cM}{{\cal M}}
\newcommand{\cF}{{\cal F}}

\newcommand{\cB}{{\cal B}}

\newcommand{\RR}{{\mathbb R}}
\newcommand{\ZZ}{{\mathbb Z}}

\def\b0{{\bf 0}}
\def\b1{{\bf 1}}
\def\be{{\bf e}}
\def\bX{{\bf X}}

\begin{document}


\title{Dependent Randomized Rounding for Matroid Polytopes \\
and Applications
}

\author{Chandra Chekuri\thanks{Dept. of Computer Science,
Univ.\ of Illinois, Urbana, IL 61801. Partially supported
by NSF grant CCF-0728782. E-mail: {\tt chekuri@cs.illinois.edu}}
\and
Jan Vondr\'ak\thanks{IBM Almaden Research Center, San Jose, CA 95120.
E-mail: {\tt jvondrak@us.ibm.com}}
\and
Rico Zenklusen\thanks{Institute for Operations Research, ETH Zurich.
E-mail: {\tt rico.zenklusen@ifor.math.ethz.ch}}}


\begin{titlepage}
\maketitle
\def\thepage {} 
\thispagestyle{empty}

\begin{abstract}
  Motivated by several applications, we consider the problem of
  randomly rounding a fractional solution in a matroid (base) polytope to
  an integral one. We consider the {\em pipage rounding} technique
  \cite{CCPV07,CCPV09,Vondrak09} and also present a new technique,
  {\em randomized swap rounding}.
  Our main technical results are concentration bounds
  for functions of random variables arising from these rounding techniques.
  We prove Chernoff-type concentration bounds for linear functions
  of random variables arising from both techniques, and also a lower-tail
  exponential bound for monotone submodular functions of variables arising
  from randomized swap rounding.

  The following are examples of our applications.
  \begin{itemize}
  \item We give a $(1-1/e-\eps)$-approximation algorithm for the
    problem of maximizing a monotone submodular function subject to 1
    matroid and $k$ linear constraints, for any constant $k \geq 1$
    and $\eps>0$.
    We also give the same
    result for a super-constant number $k$ of "loose" linear
    constraints, where the right-hand side dominates the matrix
    entries by an $\Omega(\eps^{-2} \log k)$ factor.
  \item We present a result on minimax packing problems that involve a
    matroid base constraint.  We give an $O(\log m / \log \log m)$-approximation
    for the general problem $\min \{ \lambda: \exists x \in \{0,1\}^N,x \in B(\cM),$
    $Ax \leq \lambda b \}$ where $m$ is the number of packing constraints.
    Examples include the
    low-congestion multi-path routing problem \cite{S01} and
    spanning-tree problems with capacity constraints on cuts
    \cite{BiloGRS04,HarPeled09}.
  \item We generalize the continuous greedy algorithm
    \cite{Vondrak08,CCPV09} to problems involving multiple submodular
    functions, and use it to find a $(1-1/e-\eps)$-approximate pareto
    set for the problem of maximizing a constant number of monotone
    submodular functions subject to a matroid constraint.  An example
    is the Submodular Welfare Problem where we are looking for an
    approximate pareto set with respect to individual players'
    utilities.
  \end{itemize}
\end{abstract}

\end{titlepage}
\pagenumbering{arabic}

\section{Introduction}
\label{sec:intro}
Randomized rounding is a fundamental technique introduced by Raghavan
and Thompson \cite{RaghavanT87} in order to round a fractional
solution of an LP into an integral solution. Numerous applications and
variants have since been explored and it is a standard technique in
the design of approximation algorithms and related areas.  The
original technique from \cite{RaghavanT87} (and several subsequent
papers) relies on independent rounding of the variables which allows
one to use Chernoff-Hoeffding concentration bounds for linear
functions of the variables; these bounds are critical for several
applications in packing and covering problems. However, there are many
situations in which independent rounding is not feasible due to the
presence of constraints that cannot be violated by the rounded
solution.
%
%
Various techniques are used to handle such scenarios. To name just a
few: alteration of solutions obtained by independent rounding, careful
derandomization or constructive methods when probability of a feasible
solution is non-zero but small (for example when using the Lov\'asz
Local Lemma), and various forms of correlated or dependent randomized
rounding schemes.  These methods are typically successful when one is
interested in preserving the expected value of the sum of several
random variables; the rounding schemes approximately preserve the
expected value of each random variable and then one relies on
linearity of expectation for the sum.  There are, however,
applications where one cannot use independent rounding and
nevertheless one needs concentration bounds and/or the ability to
handle non-linear objective functions such as convex or submodular
functions of the variables; the work of Srinivasan \cite{S01} and
others \cite{GKPS06,KMPS09} highlights some of these applications.
Our focus in this paper is on such schemes. In particular we consider
the problem of rounding a point in a {\em matroid polytope}
to a vertex.  We compare the existing approaches and
propose a new rounding scheme which is simple and has multiple
applications.

\medskip
\noindent
{\bf Background:}
Matroid polytopes, whose study was initiated by Edmonds in the 70's,
form one of the most important classes of polytopes associated with
combinatorial optimization problems.  (For a definition, see
Section~\ref{sec:prelims}.) Even though the full description of a
matroid polytope is exponentially large, matroid polytopes can be
optimized over, separated over, and they have strong integrality
properties such as {\em total dual integrality}. As a consequence, the
basic solution of a linear optimization problem over a matroid
polytope is always integral and no rounding is necessary.

More recently, various applications emerged where a matroid constraint
appears with additional constraints and/or the objective function is
non-linear.  In such cases, the issue of rounding a fractional solution
in the matroid polytope re-appears as a non-trivial question. One such
application is the submodular welfare problem \cite{FNW78,LLN06},
which can be formulated as a submodular maximization problem subject
to a partition matroid constraint. The rounding technique that turned
out to be useful in this context is {\em pipage rounding}
\cite{CCPV07}.

Pipage rounding was introduced by Ageev and Sviridenko \cite{AS04},
who used it for rounding fractional solutions in the bipartite
matching polytope. They used a linear program to obtain a fractional
solution to a certain problem, but the rounding procedure was based on
an auxiliary (non-linear) objective. The auxiliary objective $F(x)$
was defined in such a way that $F(x)$ would always increase or stay
constant throughout the rounding procedure. A comparison between
$F(x)$ and the original objective yields an approximation
guarantee. Calinescu et al.~\cite{CCPV07} adapted the pipage rounding
technique to problems involving a matroid constraint rather than
bipartite matchings.  Moreover, they showed that the necessary
convexity properties are satisfied whenever the auxiliary function
$F(x)$ is a {\em multilinear extension of a submodular set function} $f$.
This turned out to be crucial for further developments on submodular
maximization problems - in particular an optimal
$(1-1/e)$-approximation for maximizing a monotone submodular function
subject to a matroid constraint \cite{Vondrak08,CCPV09}, and a
$(1-1/e-\eps)$-approximation for maximizing a monotone submodular
function subject to a constant number of linear constraints
\cite{KST09}.  As one of our applications, we consider a common
generalization of these two problems.

Srinivasan \cite{S01}, and building on his work Gandhi et
al.~\cite{GKPS06}, considered dependent randomized rounding for
points in the bipartite matching polytope (and more
generally the assignment polytope); their technique can be viewed as a randomized
(and oblivious) version of pipage rounding.  The motivation for this
randomized scheme came from a different set of applications (see
\cite{S01}).  The results in \cite{S01,GKPS06} showed {\em negative
  correlation} properties for their rounding scheme which implied
concentration bounds (via \cite{PS97}) that were then useful in
dealing with additional constraints. We make some observations
regarding the results and applications in
\cite{AS04,S01,GKPS06}. Although the schemes round a point
in the assignment polytope, each constraint and objective function is
restricted to depend on a subset of the edges incident to some vertex
in the underlying bipartite graph. Further, several of the
applications in \cite{AS04,S01,GKPS06} can be naturally modeled via a
matroid constraint instead of using a bipartite graph with the above
mentioned restriction; in fact the simple partition matroid suffices.

The pipage rounding technique for matroids, as presented in
\cite{CCPV07}, is a deterministic procedure.
However, it can be randomized similarly to Srinivasan's work
\cite{S01}, and this is the variant presented in \cite{CCPV09}. This
variant starts with a fractional solution in the matroid base
polytope, $y \in B(\cM)$, and produces a random base $B \in \cM$ such
that $\E[f(B)] \geq F(y)$; here $F$ is the multilinear extension of
the submodular function $f$. A further rounding stage is needed in
case the starting point is inside the matroid polytope $P(\cM)$ rather
than the matroid base polytope $B(\cM)$; pipage rounding has been
extended to this case in \cite{Vondrak09}. In the analysis of
\cite{CCPV09,Vondrak09}, the approximation guarantees are only in
expectation. Stronger guarantees could be obtained and additional
applications would arise if we could prove {\em concentration bounds}
on the value of linear/submodular functions under such a rounding
procedure. This is the focus of this paper.

Very recently, another application has emerged where rounding in a
matroid polytope plays an essential role. Asadpour et
al.~\cite{AGMGS10} present a new approach to the Asymmetric Traveling
Salesman problem achieving an $O(\log n / \log \log n)$-approximation,
improving upon the long-standing $O(\log n)$-approximation. A crucial
step in the algorithm is a rounding procedure, which given a
fractional solution in the spanning tree polytope produces a spanning
tree satisfying certain additional constraints. The authors of
\cite{AGMGS10} use the technique of {\em maximum entropy sampling}
which gives negative correlation properties and Chernoff-type
concentration bounds for any linear function on the edges of the
graph.  Since spanning trees are bases in the {\em graphic matroid}
for any graph, this rounding procedure also falls in the framework of
randomized rounding in the matroid polytope.  However, it is not clear
whether the technique of \cite{AGMGS10} can be generalized to any
matroid or whether it could be used in applications with a submodular
objective function.

\subsection{Our work}
In this paper we study the problem of randomly rounding a
point in a matroid polytope to a vertex of the polytope.\footnote{Our results
extend easily to the case of rounding a point in the polytope of an integer
valued {\em polymatroid}. Additional applications may follow from this.}
We consider the technique of \emph{randomized pipage rounding} and also
introduce a new rounding procedure called \emph{randomized swap rounding}.
Given a starting point $x \in P(\cM)$, the procedure produces a random
independent set $S \in \cI$ such that $\Pr[i \in S] = x_i$ for each
element $i$. Our main technical results are {\em concentration bounds} for
linear {\em and} submodular functions $f(S)$ under this new rounding. We
demonstrate the usefulness of these concentration bounds via several
applications.

The randomized swap rounding procedure bears some similarity to pipage
rounding and can be used as a replacement for pipage rounding in
\cite{CCPV09,Vondrak09}.  It can be also used as a replacement for
maximum entropy sampling in \cite{AGMGS10}.  However, it has several
advantages over previous rounding procedures.  It is easy to describe
and implement, and it is very efficient.  Moreover, thanks to the
simplicity of randomized swap rounding, we are able to derive results
that are not known for previous techniques. One example is the tail
estimate for submodular functions,
Theorem~\ref{thm:swap-rounding-chernoff}. On the other hand, our
concentration bound for linear functions
(Corollary~\ref{cor:pipage-chernoff}) holds for a more general class
of rounding techniques including pipage rounding (see also
Lemma~\ref{lem:negCorr}).

Randomized swap rounding starts from an arbitrary representation of a
starting point $x \in P(\cM)$ as a convex combination of incidence
vectors of independent sets.  (This representation can be obtained by
standard techniques and in some applications it is explicitly
available.)
%
%
Once a convex representation of the starting point is obtained, the
running time of randomized swap rounding is bounded by $O(n d^2)$
calls to the membership oracle of the matroid, where $d$ is the rank
of the matroid and $n$ is the size of the ground set.
In comparison, pipage rounding performs $O(n^2)$
iterations each of which requires an expensive call to submodular
function minimization (see \cite{CCPV09}).  Maximum entropy sampling for
spanning trees in a graph $G=(V,E)$ is even more complicated; \cite{AGMGS10} does not
provide an explicit running time, but it states that the procedure
involves $O(|E|^2 |V| \log |V|)$ iterations, where in each iteration one needs to
compute a determinant (from Kirchhoff's matrix theorem) for each
edge. Also, maximum entropy sampling preserves the marginal
probabilities $\Pr[i \in S] = x_i$ only approximately, and the running
time depends on the desired accuracy.

First, we show that randomized swap rounding as well as pipage
rounding have the property that the indicator variables $X_i = [i \in
S]$ have expectations exactly $x_i$, and are {\em negatively
  correlated}.
\begin{theorem}
\label{thm:neg-correl}
Let $(x_1,\ldots,x_n) \in P(\cM)$ be a fractional solution in the
matroid polytope and $(X_1, \ldots, X_n) \in \{0,1\}^n$ an integral
solution obtained using either randomized swap rounding or randomized
pipage rounding.  Then $\E[X_i] = x_i$, and for any $T \subseteq [n]$,
\begin{inparaenum}[(i)]
\item $\E[\prod_{i \in T} X_i] \leq \prod_{i \in T} x_i$,	
\item $\E[\prod_{i \in T} (1-X_i)] \leq \prod_{i \in T} (1-x_i)$.
\end{inparaenum}
\end{theorem}

This yields Chernoff-type concentration bounds for any linear function
of $X_1, \ldots, X_n$, as proved by Panconesi and Srinivasan
\cite{PS97} (see also Theorem 3.1 in \cite{GKPS06}). Together
with Theorem~\ref{thm:neg-correl} we obtain:

\begin{corollary}
\label{cor:pipage-chernoff}
Let $a_i \in [0,1]$ and $X = \sum a_i X_i$, where $(X_1,\ldots,X_n)$
are obtained by either randomized swap rounding or randomized pipage
rounding from a starting point $(x_1,\ldots,x_n) \in P(\cM)$.
\begin{itemize}
\item If $\delta \geq 0$ and $\mu \geq \E[X] = \sum a_i x_i$, then
$ \Pr[X \geq (1+\delta) \mu] \leq
\left( \frac{e^\delta}{(1+\delta)^{1+\delta}} \right)^{\mu};$ \\
for $\delta \in [0,1]$, the bound can be simplified to
$ \Pr[X \geq (1+\delta) \mu] \leq e^{-\mu \delta^2 / 3}$.
\item If $\delta \in [0,1]$, and $\mu \leq \E[X] = \sum a_i x_i$, then
$ \Pr[X \leq (1-\delta) \mu] \leq e^{-\mu \delta^2 / 2}.$
\end{itemize}
\end{corollary}

In particular, these bounds hold for $X = \sum_{i \in S} X_i$ where
$S$ is an arbitrary subset of the variables.  We remark that in
contrast, when randomized pipage rounding is performed on bipartite
graphs, negative correlation holds only for subsets of edges incident
to a fixed vertex \cite{GKPS06}.

More generally, we consider concentration properties for a monotone
submodular function $f(R)$, where $R$ is the outcome of
randomized rounding.  Equivalently, we can also write $f(R) =
f(X_1,X_2,\ldots,X_n)$ where $X_i \in \{0,1\}$ is a random variable
indicating whether $i \in S$.  First, we consider a scenario where
$X_1,\ldots,X_n$ are {\em independent} random variables.  We prove that in
this case, Chernoff-type bounds hold for $f(X_1,X_2,\ldots,X_n)$ just
like they would for a linear function.

\begin{theorem}
\label{thm:submod-chernoff}
Let $f:\{0,1\}^n \rightarrow \RR_+$ be a monotone submodular function
with marginal values in $[0,1]$.
Let $X_1,\ldots,X_n$ be independent random variables in $\{0,1\}$.
Let $\mu = \E[f(X_1,X_2,\ldots,X_n)]$.
Then for any $\delta > 0$,
\begin{itemize}
\item
$\Pr[f(X_1,\ldots,X_n) \geq (1+\delta) \mu]
 \leq \left( \frac{e^\delta}{(1+\delta)^{1+\delta}} \right)^\mu.$
\item
$ \Pr[f(X_1,\ldots,X_n) \leq (1-\delta) \mu]
 \leq e^{-\mu \delta^2 / 2}.$
\end{itemize}
\end{theorem}

We remark that Theorem~\ref{thm:submod-chernoff} can be used to simplify
previous results for submodular maximization under linear constraints,
where variables are rounded independently \cite{KST09}.  Furthermore,
we prove a lower-tail bound in the dependent rounding case, where
$X_1,\ldots,X_n$ are produced by randomized swap rounding.

\begin{theorem}
\label{thm:swap-rounding-chernoff}
Let $f(S)$ be a monotone submodular function with marginal values in $[0,1]$,
and $F(x) = \E[f(\hat{x})]$ its multilinear extension.
Let $(x_1,\ldots,x_n) \in P(\cM)$ be a point in a matroid polytope
and $R$ a random independent set obtained from it by randomized swap rounding.
Let $\mu_0 = F(x_1,\ldots,x_n)$ and $\delta > 0$. Then $\E[f(R)] \geq \mu_0$ and
$$ \Pr[f(R) \leq (1-\delta) \mu_0] \leq e^{-\mu_0 \delta^2 / 8}.$$
\end{theorem}

We do not know how to derive this result using only the property of negative
correlations; in particular, we do not have a proof for pipage
rounding, although we suspect that a similar tail estimate
holds. (Weaker tail estimates involving a dependence on $n$ follow
directly from martingale concentration bounds; the main difficulty
here is to obtain a bound which does not depend on $n$.)  We remark
that the tail estimate is with respect to the value of the starting
point, $\mu_0 = F(x_1,\ldots,x_n)$, rather than the actual expectation
of $f(R)$, which could be larger (it would be equal for a linear
function $f$, or under independent rounding). For this reason, we do
not have an upper tail bound.  However, $\mu_0$ is the value that we
want to achieve in applications and hence this is the bound that we
need.

\medskip
\noindent
{\bf Applications:}
We next discuss several applications of our rounding scheme. While
some of the applications are concrete, others are couched in a general
framework; specific instantiations lead to various applications new
and old, and we defer some of these to a later version of the paper.
Our rounding procedure can be used to improve the running time of
some previous applications of pipage rounding \cite{CCPV09,Vondrak09}
and maximum entropy sampling \cite{AGMGS10}. In particular,
our technique significantly simplifies the algorithm and analysis
in the recent $O(\log n / \log \log n)$-approximation for the Asymmetric
Traveling Salesman problem \cite{AGMGS10}. In other applications,
we obtain approximations with high probability instead of in
expectation \cite{CCPV09,Vondrak09}. Details of these improvements
are deferred. Our new applications are as follows.

\medskip
\noindent
\emph{Submodular maximization subject to $1$ matroid and $k$ linear
  constraints.}  Given a monotone submodular function $f: 2^N
\rightarrow \RR_+$, a matroid $\cM$ on the same ground set $N$, and a
system of $k$ linear packing constraints $Ax \leq b$, we consider the
following problem: $\max \{ f(x): x \in P(\cM), Ax \le b, x \in
\{0,1\}^n \}$.  This problem is a common generalization of two
previously studied problems, monotone submodular maximization subject
to a matroid constraint \cite{CCPV09} and subject to a constant number
of linear constraints \cite{KST09}.  For any fixed $\eps > 0$ and $k
\geq 0$, we obtain a $(1-1/e-\eps)$-approximation for this problem,
which is optimal up to the arbitrarily small $\eps$ (even for 1
matroid or 1 linear constraint \cite{NW78,Feige98}), and generalizes
the previously known results in the two special cases.  We also obtain
a $(1-1/e-\eps)$-approximation when the constraints are sufficiently
"loose"; that is $b_i \geq \Omega(\eps^{-2} \log k) \cdot A_{ij}$ for
all $i,j$.

\medskip
\noindent
\emph{Minimax Integer Programs subject to a matroid constraint.}  Let
$\cM$ be a matroid on a ground set $N$ (let $n = |N|$). Let $B(\cM)$
be the base polytope of $\cM$. We consider the problem $\min \{\lambda:
Ax \le \lambda b, x \in B(\cM), x \in \{0,1\}^n\}$ where $A\in\RR^{m\times n}_+$
and $b \in \RR^n_+$.  We give an $O(\log m / \log \log m)$-approximation
for this problem, and a similar result for the min-cost version
(with given packing constraints and element costs).
This generalizes earlier results on minimax integer programs
which were considered in the context of routing and partitioning
problems \cite{RaghavanT87,LeightonLRS01,S06,S01,GKPS06}; the
underlying matroid in these settings is the partition matroid.
Another application fitting in this framework is the \emph{minimum crossing
spanning tree problem} and its geometric variant, the \emph{minimum stabbing
spanning tree problem}.
We elaborate on these in Section~\ref{sec:minimax}. 

\medskip
\noindent
\emph{Multiobjective optimization with submodular functions.}  Suppose
we are given a matroid $\cM = (N,\cI)$ and a constant number of
monotone submodular functions $f_1,\ldots,f_k:2^N \rightarrow \RR_+$.
Given a set of "target values" $V_1,\ldots,V_k$, we either find a
certificate that there is no solution $S \in \cI$ such that $f_i(S)
\geq V_i$ for all $i$, or we find a solution $S$ such that $f_i(S)
\geq (1-1/e-\eps) V_i$ for all $i$. Using the framework of
multiobjective optimization \cite{PY00}, this implies that we can find
efficiently a $(1-1/e-\eps)$-approximate pareto curve for the problem
of maximizing $k$ monotone submodular functions subject to a matroid
constraint.  A natural special case of this is the Submodular Welfare
problem, where each objective function $f_i(S)$ represents the utility
of player $i$. I.e., we can find a $(1-1/e-\eps)$-approximate pareto
curve with respect to the utilities of the $k$ players (for $k$
constant). This result involves a new variant of the
\emph{continuous greedy algorithm} from \cite{Vondrak08}, which
in some sense optimizes multiple submodular functions at the same time.
With linear objective functions $f_i$, we obtain the same
guarantees with $1-\eps$ instead of $1-1/e-\eps$.
We give more details in Section~\ref{sec:multiobj}.

\medskip \noindent

{\bf Organization:}
In Section~\ref{sec:prelims}, we present the necessary definitions.
In Section~\ref{sec:swap-rounding} the randomized swap rounding
procedure is introduced.  In Section~\ref{sec:negative-correl}, we
prove a negative correlation property for a class of rounding procedures
including randomized swap rounding and pipage rounding.  In
Section~\ref{sec:matroid+knapsacks}, we present our algorithm for
maximizing a monotone submodular function subject to $1$ matroid and
$k$ linear constraints.  In Section~\ref{sec:minimax}, we present our
results on minimax integer programs.  In Section~\ref{sec:multiobj},
we present our results on multiobjective optimization.  In
Appendix~\ref{sec:pipage-rounding}, we give a complete description of
randomized pipage rounding. In Appendix~\ref{sec:AppSwapRound}, we
present a generalization of swap rounding for rounding points
in the matroid polytope rather than the base polytope.
In Appendix~\ref{sec:submod-chernoff}, we present our
concentration bounds for submodular functions under independent
rounding, and in Appendix~\ref{sec:submod-lower-tail} our
lower-tail bound under randomized swap rounding.

\section{Preliminaries}
\label{sec:prelims}

\paragraph{Matroid polytopes.}
Given a matroid $\cM = (N, \cI)$ with rank function $r: 2^N
\rightarrow \ZZ_+$, two polytopes associated with $\cM$ are the
matroid polytope $P(\cM)$ and the matroid base polytope $B(\cM)$
\cite{Edmonds70} (see also \cite{Schrijver}). $P(\cM)$ is the convex hull
of characteristic vectors of the independent sets of $\cM$.
$$ P(\cM) = \mbox{conv} \{\b1_I: I \in \cI \} = \{ x \geq 0: \forall S; \sum_{i \in S} x_i \leq r(S) \} $$
$B(\cM)$ is the convex hull of the characteristic vectors of the
{\em bases} $\cB$ of $\cM$ , i.e. independent sets of maximum cardinality.
$$ B(\cM) = \mbox{conv} \{ \b1_B: B \in \cB \} = P(\cM) \cap \{x: \sum_{i \in N} x_i = r(N) \}.$$

\paragraph{Matroid exchange properties.}
To simplify notation, we use $+$ and $-$ for the addition
and deletion of single elements from a set, for example
$S-i+j$ denotes the set $(S\setminus \{i\})\cup \{j\}$.
The following base exchange property
of matroids is crucial in the design of our rounding algorithm.
\begin{theorem}\label{thm:strongExchange}
Let $\cM=(N,\cI)$ be a matroid and let $B_1,B_2\in \cB$. For any
$i\in B_1\setminus B_2$ there exists $j \in B_2 \setminus B_1$
such that $B_1-i+j\in \cB$ and $B_2-j+i\in \cB$.
\end{theorem}
To find an element $j$ that corresponds to a given element $i$ as
described in the above theorem, one can simply check all elements in
$B_2\setminus B_1$. Thus a corresponding element $j$ can be found by
$O(d)$ calls to an independence oracle, where $d$ is the rank of the matroid.
For many matroids, a corresponding element $j$ can be found faster. In particular, for the
graphic matroid, $j$ can be chosen to be any element $\neq i$ that
lies simultaneously in the cut defined by the connected components
of $B_1-i$ and in the unique cycle in $B_2+i$.

\paragraph{Submodular functions.}
A function $f:2^N \rightarrow \RR$ is submodular if for any $A, B
\subseteq N$, $f(A) + f(B) \ge f(A \cup B) + f(A \cap B)$.  In
addition, $f$ is monotone if $f(S) \leq f(T)$ whenever $S \subseteq
T$.  We denote by $f_A(i) = f(A+i) - f(A)$ the {\em marginal value} of
$i$ with respect to $A$. An important concept in recent work on
submodular functions \cite{CCPV07,Vondrak08,CCPV09,
  KST09,LMNS09,Vondrak09} is the {\em multilinear extension} of a
submodular function:
$$ F(x) = \E[f(x)] = \sum_{S \subseteq N} f(S) \prod_{i \in S} x_i \prod_{i \in N \setminus S}(1-x_i).$$

\paragraph{Rounding in the matroid polytope.}
A rounding procedure takes a point in the matroid polytope $x \in
P(\cM)$ and rounds it to an independent set $R \in \cI$.  In its
randomized version, it is oblivious to any objective function
and produces a random independent set, with a distribution depending
only on the starting point $x \in P(\cM)$. If the starting point is in
the matroid base polytope $B(\cM)$, the rounded solution is a (random)
base of $\cM$.

One candidate for such a rounding procedure is {\em pipage rounding}
\cite{CCPV09,Vondrak09}.  We give a complete description of the pipage
rounding technique in the appendix.  In particular, this rounding
satisfies that $\Pr[i \in R] = x_i$ for each element $i$, and
$\E[f(R)] \geq F(x)$ for any submodular function $f$ and its
multilinear extension $F$.  Our new rounding, which is described in
Section~\ref{sec:swap-rounding}, satisfies the same properties and has
additional advantages.

\section{Randomized swap rounding}\label{sec:swap-rounding}

Let $\cM=(N,\cI)$ be a matroid of rank $d = r(N)$ and let $n=|N|$.
Randomized swap rounding is a randomized procedure
that rounds a point $x\in P(\cM)$ to an independent set.
We present the procedure for points in the base polytope.  It can
easily be generalized to round any point in the matroid polytope (see
Appendix~\ref{sec:swapMP}).

Assume that $x\in B(\cM)$ is the point we want to round.  The
procedure needs a representation of $x$ as a convex combination of
bases, i.e., $x=\sum_{\ell=1}^m \beta_\ell \b1_{B_\ell}$ with
$\sum_{\ell=1}^m \beta_\ell=1, \beta_\ell \geq 0$.  Notice that by
Carath\'eodory's theorem there exists such a convex representation
using at most $n$ bases.  In some applications, the vector $x$ comes
along with a convex representation.  Otherwise, it is well-known that
one can find such a convex representation in polynomial time using the
fact that one can separate (or equivalently optimize) over the
polytope in polynomial time (see for example \cite{Schrijver98}). For
matroid polytopes, Cunningham~\cite{Cunningham84} proposed a
combinatorial algorithm that allows to find a convex representation of
$x \in B(\cM)$ using at most $n$ bases and whose runtime is bounded by
$O(n^6)$ calls to an independence oracle. In special cases, faster
algorithms are known; for example any point in the spanning tree
polytope of a graph $G=(V,E)$ can be decomposed into a convex
combination of spanning trees in $\tilde{O}(|V|^3 |E|)$ time
\cite{Gabow98}.  In general this would be the dominating
term in the running time of randomized swap rounding.

Given a convex combination of bases $x=\sum_{\ell=1}^n \beta_\ell
\b1_{B_\ell}$, the procedure takes $O(nd^2)$ calls to a matroid
independence oracle.
The rounding proceeds in $n-1$ stages, where in the first stage we
merge the bases $B_1, B_2$ (randomly) into a new base $C_2$, and
replace $\beta_1 \b1_{B_1} + \beta_2 \b1_{B_2}$ in the linear
combination by $(\beta_1+\beta_2) \b1_{C_2}$.  In the $k$-th stage,
$C_k$ and $B_{k+1}$ are merged into a new base $C_{k+1}$, and
$(\sum_{\ell=1}^k\beta_\ell) \b1_{C_k} + \beta_{k+1} \b1_{B_{k+1}}$ is
replaced in the linear combination by $(\sum_{\ell=1}^{k+1}\beta_\ell)
\b1_{C_{k+1}}$.  After $n-1$ stages, we obtain a linear combination
$(\sum_{\ell=1}^n \beta_\ell) \b1_{C_n}=\b1_{C_n}$, and the base $C_n$
is returned.

\parpic[r]{
\begin{boxedminipage}[t]{0.5\linewidth}
\vspace{-10pt}
\begin{tabbing}
\ \ \ \= \ \ \= \ \ \= \ \ \= \ \ \= \ \ \  \\
{\em Algorithm} {\bf MergeBases}$(\beta_1, B_1, \beta_2, B_2)$: \\
\> While ($B_1\neq B_2$) do\\
\> \> Pick $i\in B_1\setminus B_2$ and find $j\in B_2\setminus B_1$ such that \\
\> \> \> ~~~~~~~~~~~ $B_1-i+j\in \cI$ and $B_2-j+i\in \cI$;\\
\> \> With probability $\beta_1/(\beta_1+\beta_2)$, ~$\{B_2 \leftarrow B_2-j+i\}$;\\
\> \> ~~~~~~~~~~~~~~~~~~~~~~~~~~~~~~~~~~~Else \hspace{0.08em}~$\{ B_1 \leftarrow B_1-i+j\}$;\\
\> EndWhile \\
\> Output $B_1$.
\end{tabbing}
\end{boxedminipage}
}
The procedure we use to merge two bases, called {\bf MergeBases},
takes as input two bases $B_1$ and $B_2$ and two positive scalars
$\beta_1$ and $\beta_2$. It is described in the adjacent figure.
Notice that the procedure relies heavily on the basis exchange
property given by Theorem~\ref{thm:strongExchange} to guarantee the
existence of the elements $j$ in the while loop. As discussed in
Section~\ref{sec:prelims}, $j$ can be found by checking all elements
in $B_2\setminus B_1$.  Furthermore, since the cardinality of
$B_1\setminus B_2$ decreases at each iteration by one, the total
number of iterations is bounded by $|B_1|=d$.

\parpic[r]{
\begin{boxedminipage}[t]{0.5\linewidth}
\vspace{-10pt}
\begin{tabbing}
\ \ \ \= \ \ \= \ \ \= \ \ \= \ \ \= \ \ \  \\
{\em Algorithm} {\bf SwapRound}$(x=\sum_{\ell=1}^n \beta_\ell \b1_{B_\ell})$: \\
\> $C_1=B_1$;\\
\> For ($k=1$ to $n-1$) do\\
\> \> $C_{k+1} = ${\bf MergeBases}$(\sum_{\ell=1}^k \beta_\ell, C_k, \beta_{k+1}, B_{k+1})$;\\
\> EndFor \\
\> Output $C_{n}$.
\end{tabbing}
\end{boxedminipage}
}
The main algorithm {\bf SwapRound} is described in the figure.
It uses {\bf MergeBases} to repeatedly merge bases in the
convex decomposition of $x$.
For further analysis we present a different viewpoint on the algorithm,
namely as a random process in the matroid base polytope.
This also allows us to present the algorithm in a common framework with pipage
rounding and to draw parallels between the approaches more easily.

We denote by an \emph{elementary operation} of the swap rounding algorithm
one iteration of the while loop in the {\bf MergeBases} procedure, which is
repeatedly called in {\bf SwapRound}. Hence, an elementary operation changes
two components in one of the bases used in the convex representation of the
current point. For example, if the first elementary operation transforms
the base $B_1$ into $B_1'$, then this can be interpreted on the matroid base polytope
as transforming the point $x=\sum_{\ell=1}^n \beta_\ell \b1_{B_\ell}$ into
$\beta_1 \b1_{B_1'} + \sum_{\ell=2}^n \beta_\ell \b1_{B_\ell}$.
Hence, the {\bf SwapRound} algorithm can be seen as a sequence of $dn$
elementary operations leading to a random sequence $\bX_0, \dots, \bX_\tau$
where $\bX_t$ denotes the convex combination after $t$ elementary operations.


\section{Negative correlation for dependent rounding procedures}
\label{sec:negative-correl}

In this section, we prove a result which shows that the statement of
Theorem~\ref{thm:neg-correl} is true for a large class of random
vector-valued processes that only change at most two components at a
time. Theorem~\ref{thm:neg-correl} then easily follows by observing
that randomized swap rounding as well as pipage rounding fall in this
class of random processes.  The proof follows the same lines as
\cite{GKPS06} in the case of bipartite graphs.  The intuitive reason
for negative correlation is that whenever a pair of variables is being
modified, their sum remains constant. Hence, knowing that one variable
is high can only make the expectation of another variable lower.

\begin{lemma}\label{lem:negCorr}
  Let $\tau\in \mathbb{N}$ and let $\bX_t=(X_{1,t},\dots,X_{n,t})$ for
  $t\in \{0,\dots,\tau\}$ be a non-negative vector-valued random process
  with initial distribution given by $X_{i,0}=x_i$ with probability $1$
  $\forall i\in[n]$, and satisfying the following properties:
\begin{enumerate}
\item \label{item:mart} $\E[\bX_{t+1} \mid \bX_t]=\bX_t$ for $t\in \{0,\dots,\tau\}$ and
$i\in [n]$.
\item \label{item:2comp}$\bX_t$ and $\bX_{t+1}$ differ in at most two components for $t\in \{0,\dots, \tau-1\}$.
\item \label{item:constSum}For $t\in \{0,\dots,\tau\}$, if two components $i,j\in [n]$ change between $\bX_t$ and $\bX_{t+1}$,
then their sum is preserved: $X_{i,t+1}+X_{j,t+1}=X_{i,t}+X_{j,t}$.
\end{enumerate}
Then for any $t\in \{0,\dots,\tau\}$, the components of $\bX_t$ satisfy
$\E[\prod_{i\in S} X_{i,t}]\leq \prod_{i\in S} x_i$ $\forall S\subseteq [n]$.
\end{lemma}

\begin{proof}
We are interested in the quantity $Y_t = \prod_{i \in S} X_{i,t}$.
At the beginning of the process, we have
$\E[Y_0] = \prod_{i \in S} x_i$.
The main claim is that for each $t$, we have $\E[Y_{t+1} | \bX_t] \leq Y_t$.

Let us condition on a particular configuration of variables at time $t$,
$\bX_t=(X_{1,t}, \ldots, X_{n,t})$. We consider three cases:
\begin{itemize}
\item If no variable $X_i$, $i \in S$, is modified in step $t$,
we have $Y_{t+1} = \prod_{i \in S} X_{i,t+1} = \prod_{i \in S} X_{i,t}=Y_t$.
\item If exactly one variable $X_i$, $i \in S$, is modified in step $t$,
then by property~\ref{item:mart} of the lemma:
$$ \E[Y_{t+1} \mid \bX_t] =
 \E[X_{i,t+1} \mid \bX_t] \cdot \prod_{j \in S \setminus \{i\}} X_{j,t}
 = \prod_{j \in S} X_{j,t}=Y_t.$$
\item If two variables $X_i,X_j$, $i,j \in S$, are modified in step $t$,
we use the property that their sum is preserved:
$X_{i,t+1} + X_{j,t+1} = X_{i,t} + X_{j,t}$.
This also implies that
\begin{equation}
\label{eq:sum}
\E[(X_{i,t+1} + X_{j,t+1})^2 \mid \bX_t]
 = (X_{i,t} + X_{j,t})^2.
\end{equation}
On the other hand, the value of each variable is preserved in expectation.
Applying this to their difference, we get
$ \E[X_{i,t+1} - X_{j,t+1} \mid \bX_t]
 = X_{i,t} - X_{j,t}$. Since $\E[Z^2] \geq (\E[Z])^2$ holds for any
random variable, we get
\begin{equation}
\label{eq:diff}
\E[(X_{i,t+1} - X_{j,t+1})^2 \mid \bX_t]
 \geq (X_{i,t} - X_{j,t})^2.
\end{equation}
Combining (\ref{eq:sum}) and (\ref{eq:diff}), and using the formula
$X Y = \frac14 ((X+Y)^2 - (X-Y)^2)$, we get
$$ \E[X_{i,t+1} X_{j,t+1} \mid \bX_t]
 \leq X_{i,t} X_{j,t}. $$
Therefore,
$$ \E[Y_{t+1} \mid \bX_t] =
 \E[X_{i,t+1} X_{j,t+1} \mid \bX_t]
 \cdot \prod_{k \in S \setminus \{i,j\}} X_{k,t}
 \leq \prod_{k \in S} X_{k,t}=Y_t,$$
\end{itemize}
as claimed.
By taking expectation over all configurations $\bX_t$ we obtain $\E[Y_{t+1}] \leq \E[Y_t]$. Consequently,
$\E[\prod_{i\in S} X_{i,t}]=\E[Y_t] \leq \E[Y_{t-1}] \leq \ldots \leq \E[Y_0]=\prod_{i\in S} x_i$,
as claimed by the lemma.
\end{proof}

Any process that satisfies the conditions of Lemma~\ref{lem:negCorr} thus also satisfies
the first statement of Theorem~\ref{thm:neg-correl}. Furthermore, the second statement of
Theorem~\ref{thm:neg-correl} also follows by observing that
for any process $(X_{1,t},\dots, X_{n,t})$ that satisfies the conditions of
Lemma~\ref{lem:negCorr}, also the process $(1-X_{1,t},\dots, 1-X_{n,t})$
satisfies the conditions.
As we mentioned in Section~\ref{sec:intro}, these results imply strong
concentration bounds for linear functions of the variables
$X_1,\ldots,X_n$ (Corollary~\ref{cor:pipage-chernoff}).

Both randomized swap rounding and pipage rounding satisfy the conditions of
Lemma~\ref{lem:negCorr} (proofs can be found in the Appendix).
This implies Theorem~\ref{thm:neg-correl}.
Note that the sequences $\bX_t$ created by randomized swap rounding or pipage rounding
-- besides satisfying the conditions of Lemma~\ref{lem:negCorr} --
are Markovian, and hence they are vector-valued martingales.

\section{Submodular maximization subject to 1 matroid and k linear constraints}
\label{sec:matroid+knapsacks}

In this section, we present an algorithm for the problem of maximizing
a monotone submodular function subject to 1 matroid and $k$ linear
("knapsack") constraints.

\paragraph{Problem definition.}  {\em Given a monotone submodular
  function $f:2^N \rightarrow \RR_+$ (by a value oracle), and a
  matroid $\cM = (N, \cI)$ (by an independence oracle).  For each $i
  \in N$, we have $k$ parameters $c_{ij}$, $1 \leq j \leq k$.  A set
  $S \subseteq N$ is feasible if $S \in \cI$ and $\sum_{i \in S}
  c_{ij} \leq 1$ for each $1 \leq j \leq k$.
  The goal is to maximize $f$ over all feasible sets.
}

\medskip Kulik et al.~gave a $(1-1/e-\eps)$-approximation for the same
problem with a constant number of linear constraints, but {\em
  without} the matroid constraint \cite{KST09}.  Gupta, Nagarajan and
Ravi \cite{GuptaNR09} show that a knapsack constraint can in a
technical sense be simulated in a black-box fashion by a collection of
partition matroid constraints. Using their reduction and known results
on submodular set function maximization subject to matroid constraints
\cite{FNW78,LSV09}, they obtain a $1/(p+q+1)$-approximation with $p$
knapsacks and $q$ matroids for any $q \geq 1$ and fixed $p \geq 1$ (or
$1/(p+q+\eps)$ for any fixed $p \geq 1, q \geq 2$ and $\eps > 0$).

\subsection{Constant number of knapsack constraints}

We consider first $1$ matroid and a constant number $k$ of linear
constraints, in which case each linear constraint is thought of as a
"knapsack" constraint.  We show a $(1-1/e-\eps)$-approximation in this
case, building upon the algorithm of Kulik, Shachnai and Tamir
\cite{KST09}, which works for $k$ knapsack constraints (without a
matroid constraint).  The basic idea is that we can add the knapsack
constraints to the multilinear optimization problem
$$ \max \{ F(x): x \in P(\cM) \} $$
which is used to achieve a $(1-1/e)$-approximation for 1 matroid
constraint \cite{CCPV09}.  Using standard techniques (partial
enumeration), we get rid of all items of large value or size, and then
scale down the constraints a little bit, so that we have some room for
overflow in the rounding stage.  We can still solve the multilinear
optimization problem within a factor of $1-1/e$ and then round the
fractional solution using
randomized swap rounding (or pipage rounding). Using the fact that
randomized swap rounding makes the size in each knapsack strongly
concentrated, we conclude that our solution is feasible with constant
probability.

\paragraph{Algorithm.}
\begin{itemize}
\item Assume $0 < \eps < 1/(4k^2)$.  Enumerate all
  sets $A$ of at most $1/\eps^4$ items which form a feasible
  solution.  (We are trying to guess the most valuable items in the
  optimal solution under a greedy ordering.)  For each candidate set
  $A$, repeat the following.
\item Let $\cM' = \cM / A$ be the matroid where $A$ has been
  contracted.  For each $1 \leq j \leq k$, let $C_j = 1 - \sum_{i \in
    A} c_{ij}$ be the remaining capacity in knapsack $j$.
  Let $B$ be the set of items $i \notin A$ such that
  either $f_A(i) > \eps^4 f(A)$ or
  $c_{ij} > k \eps^3 C_j$ for some $j$
  (the item is relatively big compared to the size of some knapsack).
  Throw away all the items in $B$.
\item We consider a reduced problem on the item set $N \setminus (A
  \cup B)$, with the matroid constraint $\cM'$, knapsack capacities
  $C_j$, and objective function $g(S) = f_A(S)$.  Define a polytope
\begin{equation}
  \label{eq:P'}
  P' = \left\{ x \in P(\cM'): \forall j; \sum c_{ij} x_i \leq C_j \right\}
\end{equation}
where $P(\cM')$ is the matroid polytope of $\cM'$.
We solve (approximately) the following optimization problem:
\begin{equation}
\label{eq:multi-LP}
\max \left\{G(x): x \in (1-\eps) P' \right\}
\end{equation}
where $G(x) = \E[g(\hat{x})]$ is the multilinear extension of $g(S)$.
Since linear functions can be optimized over $P'$ in polynomial time,
we can use the continuous greedy algorithm \cite{Vondrak08} to find a
fractional solution $x^*$ within a factor of $1 - 1/e$ of optimal.
\item Given a fractional solution $x^*$, we apply randomized pipage
  rounding to $x^*$ with respect to the matroid polytope
  $P(\cM')$. Call the resulting set $R_A$.  Among all candidate sets
  $A$ such that $A \cup R_A$ is feasible, return the one maximizing
  $f(A \cup R_A)$.
\end{itemize}

We remark that the value of this algorithm (unlike the
$(1-1/e)$-approximation for 1 matroid constraint) is purely
theoretical, as it relies on enumeration of a huge (constant) number
of elements.

\begin{theorem}
\label{thm:matroid+knapsacks}
  With constant positive probability, the algorithm above returns a solution
  of value at least $(1-1/e-3\eps) OPT$.
\end{theorem}

\begin{proof}
  Consider an optimum solution $O$, i.e. $OPT = f(O)$.  Order the
  elements of $O$ greedily by decreasing marginal values, and let $A
  \subseteq O$ be the elements whose marginal value is at least
  $\eps^4 OPT$. There can be at most $1/\eps^4$ such elements,
  and so the algorithm will consider them as one of the candidate
  sets. We assume in the following that this is the set $A$ chosen by
  the algorithm.

  We consider the reduced instance, where $\cM' = \cM / A$ and the
  knapsack capacities are $C_j = 1 - \sum_{i \in A} c_{ij}$. $O
  \setminus A$ is a feasible solution for this instance and we have
  $g(O \setminus A) = f_A(O \setminus A) = OPT - f(A)$.  We know that
  in $O \setminus A$, there are no items of marginal value more than
  the last item in $A$.  In particular, $f_A(i) \leq \eps^4 f(A) \leq
  \eps^4 OPT$ for all $i \in O \setminus A$. We throw away all items
  where $f_A(i) > \eps^4 f(A)$ but this does not affect any item in $O
  \setminus A$.  We also throw away the set $B \subseteq N \setminus
  A$ of items whose size in some knapsack is more then $k \eps^3 C_j$.
  In $O \setminus A$, there can be at most $1/(k\eps^3)$ such items
  for each knapsack, i.e. $1 / \eps^3$ items in total. Since their
  marginal values with respect to $A$ are bounded by $\eps^4 OPT$,
  these items together have value $g(O \cap B) = f_A(O \cap B) \leq
  \eps OPT$.  $O' = O \setminus (A \cup B)$ is still a feasible set
  for the reduced problem, and using submodularity, its value is
$$ g(O') = g((O \setminus A) \setminus (O \cap B)) \geq g(O \setminus A) - g(O \cap B)
 \geq OPT - f(A) - \eps OPT.$$

 Now consider the multilinear problem (\ref{eq:multi-LP}). Note that
 the indicator vector $\b1_{O'}$ is feasible in $P'$, and hence
 $(1-\eps) \b1_{O'}$ is feasible in $(1-\eps) P'$.  Using the
 concavity of $G(x)$ along the line from the origin to $\b1_{O'}$, we
 have $G((1-\eps) \b1_{O'}) \geq (1-\eps) g(O') \geq
 (1-2\eps) OPT - f(A).$ Using the continuous greedy algorithm
 \cite{Vondrak08}, we find a fractional solution $x^*$ of value
$$ G(x^*) \geq (1-1/e) G((1-\eps) \b1_{O'}) \geq (1-1/e-2\eps) OPT - f(A).$$

Finally, we apply randomized swap rounding (or pipage rounding) to $x^*$ and call the
resulting set $R$.  By the construction of randomized swap rounding, $R$ is
independent in $\cM'$ with probability $1$.
However, $R$ might violate some of the knapsack constraints.

Consider a fixed knapsack constraint, $\sum_{i \in S} c_{ij} \leq
C_j$.  Our fractional solution $x^*$ satisfies $\sum c_{ij} x^*_i \leq
(1-\eps) C_j$.  Also, we know that all sizes in the reduced instance
are bounded by $c_{ij} \leq k \eps^3 C_j$.  By scaling, $c'_{ij} =
c_{ij} / (k \eps^3 C_j)$, we can apply
Corollary~\ref{cor:pipage-chernoff} with $\mu = (1-\eps) / (k
\eps^3)$:
$$ \Pr[ \sum_{i \in R} c_{ij} > C_j] \leq \Pr[ \sum_{i \in R} c'_{ij} > (1+\eps) \mu]
 \leq e^{-\mu \eps^2 / 3} < e^{-1/4k\eps}. $$
On the other hand, consider the objective function $g(R)$.
In the reduced instance, all items have value $g(i) \leq \eps^4 OPT$.
Let $\mu = G(x^*) / (\eps^4 OPT)$. Then, Theorem~\ref{thm:swap-rounding-chernoff} implies
$$ \Pr[g(R) \leq (1-\delta) G(x^*)] = \Pr[f(R) / (\eps^4 OPT) \leq (1-\delta) \mu]
 \leq e^{-\delta^2 \mu / 8} = e^{-\delta^2 G(x^*) / 8 \eps^4 OPT}.$$
We set $\delta = \frac{OPT}{G(x^*)} \eps$ and obtain
$$ \Pr[g(R) \leq G(x^*) - \eps OPT] \leq e^{-OPT / 8\eps^2 G(x^*)} \leq e^{-1/8\eps^2}.$$
By the union bound,
$$ \Pr[g(R) \leq G(x^*) - \eps OPT \mbox{ or } \exists j; \sum_{i \in R} c_{ij} > C_j]
 \leq e^{-1/8\eps^2} + k e^{-1/4k\eps}.$$
For $\eps < 1/(4k^2)$, this probability is at most $e^{-2k^4} + k e^{-k} < 1$.
If this event does not occur, we have a feasible solution of value
$f(R) = f(A) + g(R) \geq f(A) + G(x^*) - \eps OPT \geq (1-1/e-3\eps) OPT$.

\end{proof}

\subsection{Loose packing constraints}

In this section we consider the case when the number of linear
packing constraints is not a fixed constant. The notation we use in this case
is that of a packing integer program:
$$\max \{ f(x): x \in P(\cM), Ax \le b, x \in \{0,1\}^n \}.$$
Here $f: 2^N \rightarrow \RR$ is a monotone submodular function with
$n = |N|$, $\cM = (N,\cI)$ is a matroid,
$A \in \RR_+^{k \times n}$ is a non-negative matrix and $b \in \RR_+^k$
is a non-negative vector. This problem has been studied
extensively when $f(x)$ is a linear function, in other words $f(x) = w^T x$
for some non-negative weight vector $w \in \RR^n$. Even this case
with $A,b$ having only $0,1$ entries captures the maximum independent
set problem in graphs and hence is NP-hard to approximate to within an
$n^{1-\eps}$-factor for any fixed $\eps>0$.
For this reason a variety of restrictions on $A,b$ have been studied.

We consider the case when the constraints are sufficiently loose,
i.e. the right-hand side $b$ is significantly larger than entries in
$A$: in particular, we assume $b_i \ge c \log k \cdot \max_{j}
A_{ij}$ for $1 \le i \le k$.  In this case, we propose a
straightforward algorithm which works as follows.

\paragraph{Algorithm.}
\begin{itemize}
\item Let $\eps = \sqrt{6/c}$.
Solve (approximately) the following optimization problem:
$$ \max \{ F(x) : x \in (1-\eps) P \} $$
where $F(x) = \E[f(\hat{x})]$ is the multilinear extension of $f(S)$, and
\begin{equation*}
P = \{x \in P(\cM) \mid \forall i; \sum_{j\in N} A_{ij} x_j \leq b_i \}.
\end{equation*}
Since linear functions can be optimized over $P$ in polynomial time,
we can use the continuous greedy algorithm \cite{Vondrak08} to find a
fractional solution $x^*$ within a factor of $1 - 1/e$ of optimal.
\item Apply randomized pipage rounding to $x^*$ with respect to the matroid polytope $P(\cM)$.
  If the resulting solution $R$ satisfies the packing constraints, return $R$;
  otherwise, fail.
\end{itemize}

\begin{theorem}
Assume that $A \in \RR^{k \times n}$ and $b \in \RR^k$ such that
$b_i \geq A_{ij} c \log k$ for all $i,j$ and some constant $c = 6/\eps^2$.
Then the algorithm above gives a $(1-1/e-O(\eps))$-approximation with constant probability.
\end{theorem}

We remark that it is NP-hard to achieve a better than $(1-1/e)$-approximation even when
$k=1$ and the constraint is very loose ($A_{ij} = 1$ and $b_i \rightarrow \infty$) \cite{Feige98}.

\begin{proof}
The proof is similar to that of Theorem~\ref{thm:matroid+knapsacks}, but simpler.
We only highlight the main differences.

In the first stage we obtain a fractional solution such that $F(x^*) \geq (1-\eps)(1-1/e) OPT$.
Randomized swap rounding yields a random solution $R$ which satisfies the matroid constraint.
It remains to check the packing constraints. For each $i$, we have
$$ \E[ \sum_{j \in R} A_{ij} ] = \sum_{j \in N} A_{ij} x^*_j \leq (1-\eps) b_i.$$
The variables $X_j$ are negatively correlated and by Corollary~\ref{cor:pipage-chernoff}
with $\delta = \eps = \sqrt{6/c}$ and $\mu = c \log k$,
$$ \Pr[ \sum_{j \in R} A_{ij} > b_i] < e^{-\delta^2 \mu / 3} = \frac{1}{k^2}.$$
By the union bound, all packing constraints are satisfied with probability at least $1 - 1/k$.
We assume here that $k = \omega(1)$.
By using Theorem~\ref{thm:swap-rounding-chernoff}, we can also conclude that
the value of the solution is at least $(1 - 1/e - O(\eps)) OPT$ with constant probability.
\end{proof}

\section{Minimax integer programs with a matroid constraint}
\label{sec:minimax}

Minimax integer programs are motivated by applications to routing
and partitioning. The setup is as follows; we follow \cite{S06}.
We have boolean variables $x_{i,j}$ for $i \in [p]$ and $j \in [\ell_i]$
for integers $\ell_1, \ldots,\ell_p$. Let $n = \sum_{i \in [p]} \ell_i$.
The goal is to minimize $\lambda$ subject to:
\begin{itemize}
\item equality constraints: $\forall i \in [p], \sum_{j \in [\ell_i]} x_{i,j} = 1$
\item a system of linear inequalities $A x \le \lambda \b1$ where $A \in [0,1]^{m \times n}$
\item integrality constraints: $x_{i,j} \in \{0,1\}$ for all $i,j$.
\end{itemize}

The variables $x_{i,j}$, $j \in [\ell_i]$ for each $i \in [p]$ capture
the fact that exactly one option amongst the $\ell_i$ options in group
$i$ should be chosen. A canonical example is the congestion
minimization problem for integral routings in graphs where for each
$i$, the $x_{i,j}$ variables represent the different paths for routing
the flow of a pair $(s_i,t_i)$ and the matrix $A$ encodes the capacity
constraints of the edges. A natural approach is to solve the natural
LP relaxation for the above problem and then apply randomized rounding
by choosing independently for each $i$ exactly one $j \in [\ell_i]$
where the probability of choosing $j \in [\ell_i]$ is exactly equal to
$x_{i,j}$. This follows the randomized rounding method of Raghavan
and Thompson for congestion minimization \cite{RaghavanT87} and one
obtains an $O(\log m/\log \log m)$-approximation with respect to the
fractional solution. Using Lov\'asz Local Lemma (and complicated
derandomization) it is possible to obtain an improved bound of $O(\log
q/\log \log q)$ \cite{LeightonLRS01,S06} where $q$ is the maximum
number of non-zero entries in any column of $A$. This refined bound
has various applications.

Interestingly, the above problem becomes non-trivial if we make a
slight change to the equality constraints.
Suppose for each $i \in [p]$ we now have an equality constraint of the form
$\sum_{j \in [\ell_i]} x_{i,j} = k_i$ where $k_i$ is an integer. For routing,
this corresponds to a requirement of $k_i$ paths for pair $(s_i,t_i)$.
Now the standard randomized rounding doesn't quite work for this
{\em low congestion multi-path routing problem}. Srinivasan
\cite{S01}, motivated by this generalized routing problem, developed
dependent randomized rounding and used the negative correlation
properties of this rounding to obtain an $O(\log m/\log \log m)$-approximation.
This was further generalized in \cite{GKPS06} as
randomized versions of pipage rounding in the
context of other applications.

\subsection{Congestion minimization under a matroid base constraint}
\label{sec:minimax-unweighted}

Here we show that our dependent rounding in matroids allows a
clean generalization of the type of constraints considered in several
applications in \cite{S01,GKPS06}. Let $\cM$ be a matroid on a ground
set $N$. Let $B(\cM)$ be the base polytope of $\cM$. We
consider the problem
$$ \min \left\{\lambda: \exists x \in \{0,1\}^N, x \in B(\cM), Ax \le \lambda \b1 \right\} $$
where $A \in [0,1]^{m \times N}$.
We observe that the previous problem with the variables partitioned
into groups and equality constraints can be cast naturally
as a special case of this matroid constraint problem; the equality
constraints simply correspond to a partition matroid on the
ground set of all variables $x_{i,j}$.

However, our framework is much more flexible. For example, consider the spanning tree
problem with packing constraints: each edge has a weight $w_e$ and we want to
minimize the maximum load on any vertex, $\max_{v \in V} \sum_{e \in \delta(v)} w_e$.
This problem also falls within our framework.

\begin{theorem}\label{thm:minCongestion}
There is an $O(\log m / \log \log m)$-approximation for the problem
$$ \min \left\{\lambda: \exists x \in \{0,1\}^N, x \in B(\cM), Ax \le \lambda \b1 \right\}, $$
where $m$ is the number of packing constraints, i.e. $A \in [0,1]^{m \times N}$.
\end{theorem}

\begin{proof}
  Fix a value of $\lambda$. Let $Z(\lambda) = \{ j \mid \exists i;
  A_{ij} > \lambda \}$.  We can force $x_j = 0$ for all $j \in
  Z(\lambda)$, because no element $j \in Z(\lambda)$ can be in a
  feasible solution for $\lambda$.  In polynomial time, we can check
  the feasibility of the following LP:
$$ P_\lambda =  \left\{ x \in B(\cM): Ax \le \lambda \b1, x|_{Z(\lambda)} = 0 \right\} $$
(because we can separate over $B(\cM)$ and the additional packing
constraints efficiently).  By binary search, we can find (within
$1+\eps$) the minimum value of $\lambda$ such that $P_\lambda \neq
\emptyset$.  This is a lower bound on the actual optimum
$\lambda_{OPT}$.  We also obtain the corresponding fractional solution
$x^*$.

We apply randomized swap rounding (or randomized pipage rounding) to
$x^*$, obtaining a random set $R$.  $R$ satisfies the matroid base
constraint by definition.  Consider a fixed packing constraint (the
$i$-th row of $A$). We have
$$ \sum_{j\in N} A_{ij} x^*_j \leq \lambda $$
and all entries $A_{ij}$ such that $x^*_j > 0$ are bounded by $\lambda$.
We set $\tilde{A}_{ij} = A_{ij} / \lambda$, so that we can use Corollary~\ref{cor:pipage-chernoff}.
We get
$$ \Pr[ \sum_{j \in R} A_{ij} > (1+\delta) \lambda ] = \Pr[ \sum_{j \in R} \tilde{A}_{ij} > 1+\delta ]
 < \left( \frac{e^{\delta}}{(1+\delta)^{1+\delta}} \right)^{\mu}.$$
For $\mu = 1$ and $1+\delta = \frac{4 \log m}{\log \log m}$, this probability is bounded by
$$ \Pr[ \sum_{j \in R} A_{ij} > (1+\delta) \lambda ] \leq
\left( \frac{e \log \log m}{4 \log m} \right)^{\frac{4 \log m}{\log
    \log m}} < \left( \frac{1}{\sqrt{\log m}} \right)^{\frac{4 \log
    m}{\log \log m}} = \frac{1}{m^2} $$ for sufficiently large
$m$. Therefore, all $m$ constraints are satisfied within a factor of
$1+\delta = \frac{4 \log m}{\log \log m}$ with high probability.
\end{proof}

We remark that the approximation guarantee can be made an "almost
additive" $O(\log m)$, in the following sense: Assuming that the
optimum value is $\lambda^*$, for any fixed $\eps>0$ we can find a
solution of value $\lambda \leq (1+\eps) \lambda^* + O(\frac{1}{\eps}
\log m)$.  Scaling is important here: recall that we assumed $A \in
[0,1]^{N \times m}$.  We omit the proof, which follows by a similar
application of the Chernoff bound as above, with $\mu = \lambda^*$ and
$\delta = \eps + O(\frac{1}{\eps \lambda^*} \log m)$.

\paragraph{Minimum Stabbing and Crossing Tree Problems:}
Another interesting application of Theorem~\ref{thm:minCongestion}, is
to the minimum stabbing and crossing tree problems.  Bilo et al.~\cite{BiloGRS04},
motivated by several applications, considered the
crossing spanning tree problem. The input is a graph $G=(V,E)$ and an
explit set ${\cal C}$ of $m$ cuts in $G$. The goal is to find a
spanning tree that minimizes the number of edges crossing any cut in
${\cal C}$. The algorithm in \cite{BiloGRS04} returns a tree that
crosses any cut in ${\cal C}$ at most $O((\log m + \log n)(\gamma^*
+\log n))$ times where $\gamma^*$ is the optimal solution value; the
authors claim an improved bound of $O(\gamma^* \log n + \log m)$ in a
subsequent version of the paper.

The minimum stabbing tree problem arises in computational geometry:
the input is a set $V=\{v_1,\dots,v_n\}$ of points in $\mathbb{R}^d$;
it is assumed that $d$ is a constant and the case of $2$-dimensions is
of particular interest.  The task is to construct a spanning tree on
$V$ by connecting vertices with straight lines such that the crossing
number, which is the maximum number of edges that are intersected by
any hyperplane, is minimized. This problem was shown to be NP-hard by
Fekete et al.~\cite{Fekete04}.  It is relatively easy to see that the
stabbing tree problem is a special case of the crossing spanning
tree problem; the number of combinatorially distinct cuts induced by
the hyperplanes is $O(n^d)$, one for each set of $d$ points that
define a hyperplane through them.  Thus, the result in
\cite{BiloGRS04} implies that there is an algorithm for the stabbing
tree problem that returns a tree with crossing number $O(\lambda^*
\log n)$ where $\lambda^*$ is the tree with the smallest crossing
number (note that this is via the improved bound claimed by the
authors of \cite{BiloGRS04} in a longer version). Unaware of the work
in \cite{BiloGRS04}, HarPeled very recently \cite{HarPeled09} gave a
polynomial time algorithm for the stabbing tree problem that outputs a
tree with crossing number $O(\lambda^* \log n + \log^2 n/\log \log
n)$.

Both of the above problems can be cast as special cases of the
minimization problem presented in Theorem~\ref{thm:minCongestion},
where $\cM$ is the graphic matroid and each row of $A$ corresponds to
the incidence vector of a cut.  Theorem~\ref{thm:minCongestion}
implies that using dependent randomized rounding, an $O(\log n
/\log\log n)$-approximation can be obtained for the stabbing tree
problem and an $O(\log m/\log \log m)$-approximation for the crossing
spanning tree problem. The approximation guarantee can be transformed
into an almost additive one as well, leading to a solution of value
$\lambda \leq (1+\eps) \lambda^*+O(\frac{1}{\eps}\log n)$ for the
stabbing tree problem and a solution of value $\gamma \leq (1+\eps)
\gamma^*+O(\frac{1}{\eps}\log m)$ for the crossing spanning tree
problem. Note that these additive results imply a constant factor
approximation if the optimal value is $\Omega(\log n)$ and
$\Omega(\log m)$ respectively.

We remark that the results we obtain for the above problems can also
be obtained by the maximum entropy sampling approach for
spanning trees from \cite{AGMGS10}; our algorithms have the advantage
of being simpler and more efficient.

\subsection{Min-cost matroid bases with packing constraints}
\label{sec:minimax-weighted}

We can similarly handle the case where in addition we want to minimize
a linear objective function.  An example of such a problem would be a
multi-path routing problem minimizing the total cost in addition to
congestion. Another example is the minimum-cost spanning tree with
packing constraints for the edges incident with each vertex. We remark
that in case the packing constraints are simply degree bounds, strong
results are known - namely, there is an algorithm that finds a
spanning tree of optimal cost and violating the degree bounds by at
most one \cite{SL07}. In the general case of finding a matroid base
satisfying certain "degree constraints", there is an algorithm
\cite{KLS08} that finds a base of optimal cost and violating the
degree constraints by an additive error of at most $\Delta - 1$, where
each element participates in at most $\Delta$ constraints
(e.g. $\Delta=2$ for degree-bounded spanning trees).  The algorithm of
\cite{KLS08} also works for upper and lower bounds, violating each
constraint by at most $2 \Delta - 1$. See \cite{KLS08} for more
details.

We consider a variant of this problem where the packing constraints
can involve arbitrary weights and capacities.  We show that we can
find a matroid base of near-optimal cost which violates the packing
constraints by a multiplicative factor of $O(\log m / \log \log m)$,
where $m$ is the total number of packing constraints.

\begin{theorem}
  There is a $(1+\eps,O(\log m / \log \log m))$-bicriteria
  approximation for the problem
$$ \min \left\{ c^T x : x \in \{0,1\}^N, x \in B(\cM), Ax \le b \right\}, $$
where $A \in [0,1]^{m \times N}$ and $b \in \RR^N$; the first
guarantee is w.r.t.~the cost of the solution and the second guarantee
w.r.t.~the overflow on the packing constraints.
\end{theorem}

\begin{proof}
  We give a sketch of the proof. First, we throw away all elements
  that on their own violate some packing constraint. Then, we solve
  the following LP:
  $$ \min \left\{ c^T x : x \in B(\cM), Ax \leq b \right\}.$$
  Let the optimum solution be $x^*$. We apply randomized swap rounding
  (or randomized pipage rounding) to $x^*$, yielding a random solution
  $R$.  Since each of the $m$ constraints is satisfied in expectation,
  and each element alone satisfies each packing constraint, we get by
  the same analysis as above that with high probability, $R$ violates
  every constraint by a factor of $O(\log m / \log \log m)$.

  Finally, the expected cost of our solution is $c^T x^* \leq OPT$.
  By Markov's inequality, the probability that $c(R) > (1+\eps) OPT$
  is at most $1 / (1+\eps) \leq 1 - \eps/2$. With probability at least
  $\eps/2 - o(1)$, $c(R) \leq (1+\eps) OPT$ and all packing
  constraints are satisfied within $O(\log m / \log \log m)$.
\end{proof}

Let us rephrase this result in the more familiar setting of spanning
trees.  Given packing constraints on the edges incident with each
vertex, using arbitrary weights and capacities, we can find a spanning
tree of near-optimal cost, violating each packing constraint by a
multiplicative factor of $O(\log m / \log \log m)$.  As in the
previous section, if we assume that the weights are in $[0,1]$, this
can be replaced by an additive factor of $O(\frac{1}{\eps} \log m)$
while making the multiplicative factor $1+\eps$ (see the end of
Section~\ref{sec:minimax-unweighted}).

In the general case of matroid bases, our result is incomparable to
that of \cite{KLS08}, which provides an additive guarantee of
$\Delta-1$. (The assumption here is that each element participates in
at most $\Delta$ degree constraints; in our framework, this
corresponds to $A \in \{0,1\}^{m \times N}$ with $\Delta$-sparse
columns.)  When elements participate in many degree constraints
($\Delta \gg \log m$) and the degree bounds are $b_i = O(\log m)$, our
result is actually stronger in terms of the packing constraint
guarantee.

\paragraph{Asymmetric Traveling Salesman and Maximum Entropy Sampling:}
In a recent breakthrough, \cite{AGMGS10} obtained an $O(\log n/\log
\log n)$-approximation for the ATSP problem. A crucial ingredient in
the approach is to round a point $x$ in the spanning tree
polytope to a tree $T$ such that no cut of $G$ contains too many edges of $T$,
and the cost of the tree is within a constant factor of the cost of $x$.
For this purpose, \cite{AGMGS10} uses the maximum entropy sampling approach which
also enjoys negative correlation properties and hence one can get
Chernoff-type bounds for linear sums of the variables; moreover $T$ contains
each edge $e$ with probability $x_e$.  We note
that the number of cuts is exponential in $n$. To address this issue,
\cite{AGMGS10} uses Karger's result on the number of cuts in a graph
within a certain weight range: assuming that the minimum cut is at least $1$,
there are only $O(n^{2\alpha})$ cuts of weight in $(\alpha/2, \alpha]$
for any $\alpha \geq 1$. Maximum entropy sampling is technically quite involved
and also computationally expensive.
Our rounding procedures can be used
in place of maximum entropy sampling to simplify the algorithm and
the analysis in \cite{AGMGS10}.

\section{Multiobjective optimization with submodular functions}
\label{sec:multiobj}

In this section, we consider the following problem: Given a matroid
$\cM = (N,\cI)$ and $k$ monotone submodular functions
$f_1,\ldots,f_k:2^N \rightarrow \RR_+$, in what sense can we maximize
$f_1(S),\ldots,f_k(S)$ simultaneously over $S \in \cI$?  This question
has been studied in the framework of {\em multiobjective
  optimization}, popularized in the CS community by the work of
Papadimitriou and Yannakakis \cite{PY00}.  The set of all solutions
which are optimal with respect to $f_1(S),\ldots,f_k(S)$ is captured
by the notion of a {\em pareto set}: the set of all solutions $S$ such
that for any other feasible solution $S'$, there exists $i$ for which
$f_i(S') < f_i(S)$.  Since the pareto set in general can be
exponentially large, we settle for the notion of a $\eps$-approximate
pareto set, where the condition is replaced by $f_i(S') < (1+\eps)
f_i(S)$. Papadimitriou and Yannakakis show the following equivalence
\cite[Theorem 2]{PY00}:

\begin{proposition}\label{prop:PY}
An $\eps$-approximate pareto set can be found in polynomial time, if
and only if the following problem can be solved: Given
$(V_1,\ldots,V_k)$, either return a solution with $f_i(S) \geq V_i$
for all $i$, or answer that there is no solution such that $f_i(S)
\geq (1+\eps) V_i$ for all $i$.
\end{proposition}

The latter problem is exactly what we address in this section. We show
the following result.

\begin{theorem}
\label{thm:submod-pareto}
  For any fixed $\eps>0$ and $k \geq 2$, given a matroid $\cM =
  (N,\cI)$, monotone submodular functions $f_1,\ldots,f_k: 2^N
  \rightarrow \RR_+$, and values $V_1,\ldots,V_k \in \RR_+$, in
  polynomial time we can either
\begin{itemize}
\item find a solution $S \in \cI$ such that $f_i(S) \geq (1-1/e-\eps)
  V_i$ for all $i$, or
\item return a certificate that there is no solution with $f_i(S) \geq
  V_i$ for all $i$.
\end{itemize}
If $f_i(S)$ are linear functions, the guarantee in the first case
becomes $f_i(S) \geq (1-\eps) V_i$.
\end{theorem}

This together with Proposition~\ref{prop:PY} implies that for any constant
number of linear objective functions subject to a matroid constraint, an
$\eps$-approximate pareto set can be found in polynomial time.  (This
was known in the case of multiobjective spanning trees \cite{PY00}.)
Furthermore, a straightforward modification of Prop.~\ref{prop:PY}
(see \cite{PY00}, Theorem 2) implies that for monotone submodular
functions $f_i(S)$, we can find a $(1-1/e-\eps)$-approximate pareto
set.

Our algorithm requires a modification of the continuous greedy algorithm
from \cite{Vondrak08,CCPV09}. We show the following, which might be
useful in other applications as well. In the following lemma,
we do not require $k$ to be constant.

\begin{lemma}
\label{lemma:multi-cont-greedy}
Consider monotone submodular functions $f_1,\ldots,f_k:2^N \rightarrow \RR_+$,
their multilinear extensions $F_i(x) = \E[f_i(\hat{x})]$
and a down-monotone polytope $P \subset \RR_+^N$ such that we can optimize
linear functions over $P$ in polynomial time. Then given $V_1,\ldots,V_k \in \RR_+$
we can either
\begin{itemize}
\item find a point $x \in P$ such that $F_i(x) \geq (1-1/e) V_i$ for all $i$, or
\item return a certificate that there is no point $x \in P$ such that $F_i(x) \geq V_i$ for all $i$.
\end{itemize}
\end{lemma}

\begin{proof}
We refer to Section 2.3 of \cite{CCPV09} for intuition and notation.
Assuming that there is a solution $S \in \cI$
achieving $f_i(S) \geq V_i$, Section 2.3 in \cite{CCPV09} implies that
for any fractional solution $y \in P(\cM)$ there is a direction
$v^*(y) \in P(\cM)$ such that $v^*(y) \cdot \nabla F_i(y) \geq V_i -
F_i(y)$.  Moreover, the way this direction is constructed is by going
towards the actual optimum - i.e., this direction is the same for all
$i$. Assuming that such a direction exists, we can find it by linear
programming. If the LP is infeasible, we have a certificate that there
is no solution satisfying $f_i(S) \geq V_i$ for all $i$.  Otherwise,
we follow the continuous greedy algorithm and the analysis implies
that
$$ \frac{dF_i}{dt} \geq v^*(y(t)) \cdot \nabla F(y(t)) \geq V_i - F_i(y(t)) $$
which implies $F_i(y(1)) \geq (1-1/e) V_i$.
\end{proof}

Given Lemma~\ref{lemma:multi-cont-greedy}, we sketch the proof of
Theorem~\ref{thm:submod-pareto} as follows. First, we guess a constant
number of elements so that for each remaining element $j$, the
marginal value for each $i$ is at most $\eps^3 V_i$.  In the
following, we just assume that $f_i(j) \leq \eps^3 V_i$ for all $i,j$.
For each objective function $f_i$, we consider the multilinear
relaxation of the problem:
$$ \max \{ F_i(x): x \in P(\cM) \} $$
where $F_i(x) = \E[f_i(\hat{x})]$.  We apply
Lemma~\ref{lemma:multi-cont-greedy} to find a fractional solution
$y^*$ satisfying $F_i(y^*) \geq (1-1/e) V_i$ for all $i$ (or a
certificate that there is no solution $y \in P(\cM)$ such that $F_i(y)
\geq V_i$ for all $i$; this implies that there is no feasible solution
$S$ such that $f_i(S) \geq V_i$ for all $i$).  For linear objective
functions, the problem is much simpler: then $F_i(x)$ are linear
functions and we can find a fractional solution satisfying $F_i(y^*)
\geq V_i$ directly by linear programming.

We apply randomized swap rounding to $y^*$, to obtain a random
solution $R \in \cI$ satisfying the lower-tail concentration bound of
Theorem~\ref{thm:swap-rounding-chernoff}. The marginal values of $f_i$
are bounded by $\eps^3 V_i$, so by standard scaling we obtain
$$ \Pr[f_i(R) < (1-\delta) F_i(y^*)] < e^{-\delta^2 F_i(y^*) / 8 \eps^3 V_i}
\leq e^{-\delta^2 / 16 \eps^3}. $$ Hence, we can set $\delta = \eps$
and obtain error probability at most $e^{-1 / 16 \eps}$.  By the union
bound, the probability that $f_i(R) < (1-\eps) F_i(y^*)$ for any $i$
is at most $k e^{-1 / 16 \eps}$. For sufficiently small $\eps>0$, this
is a constant probability smaller than $1$. Then, $f_i(R) \geq
(1-1/e-\eps) V_i$ for all $i$. This proves Theorem~\ref{thm:submod-pareto}.

To conclude, we are able to find a $(1-1/e-\eps)$-approximate
pareto set for any constant number of monotone submodular functions
and any matroid constraint.  This has a natural interpretation
in the setting of the Submodular Welfare Problem (which is
a special case, see \cite{FNW78,LLN06}). Then each objective function
$f_i(S)$ is the utility function of a player, and we want to find a
pareto set with respect to all possible allocations. To summarize, we
can find a set of all allocations that are not dominated by any other
allocation within a factor of $1-1/e-\eps$ per player.

\paragraph{Acknowledgments:} CC and JV thank Anupam Gupta
for asking about the approximability of maximizing a monotone
submodular set function subject to a matroid constraint and a constant
number of knapsack constraints; this motivated their work.  JV thanks
Ilias Diakonikolas for fruitful discussions concerning multiobjective
optimization that inspired the application in
Section~\ref{sec:multiobj}.  RZ is grateful to Michel Goemans for
introducing him to a version of Shannon's switching game that inspired
the randomized swap rounding algorithm. We thank Mohit Singh for pointing
out \cite{BiloGRS04,HarPeled09}.


\appendix

\section{Randomized pipage rounding}
\label{sec:pipage-rounding}

Let us summarize the pipage rounding technique in the context of matroid polytopes \cite{CCPV07,CCPV09}.
The basic version of the technique assumes that we start with a point in the matroid base
polytope, and we want to round it to a vertex of $B(\cM)$. In each step, we have a fractional
solution $y \in B(\cM)$ and a {\em tight set} $T$ (satisfying $y(T) = r(T)$) containing at least
two fractional variables. We modify the two fractional variables in such a way that their sum
remains constant, until some variable becomes integral or a new constraint becomes tight.
If a new constraint becomes tight, we continue with a new tight set, which can be shown to be
a proper subset of the previous tight set \cite{CCPV07,CCPV09}.
Hence, after $n$ steps we produce a new integral variable, and the process terminates after $n^2$ steps.

In the randomized version of the technique, each step is randomized in such a way
that the expectation of each variable is preserved.
Here is the randomized version of pipage rounding \cite{CCPV09}.
The subroutine {\bf HitConstraint}$(y,i,j)$ starts from $y$ and tries to increase $y_i$
and decrease $y_j$ at the same rate, as long as the the solution is inside $B(\cM)$.
It returns a new point $y$ and a tight set $A$, which would be violated if we go any further.
This is used in the main algorithm {\bf PipageRound}$(\cM,y)$, which repeats the process
until an integral solution in $B(\cM)$ is found.

\vspace{-10pt}
\begin{tabbing}
\ \ \ \= \ \ \= \ \ \= \ \ \= \ \ \= \ \ \  \\
{\em Subroutine} {\bf HitConstraint}($y$, $i$, $j$): \\
\> Denote ${\cal A} = \{A \subseteq X: i \in A, j \notin A \}$; \\
\> Find $\delta = \min_{A \in {\cal A}} (r_\cM(A)-y(A))$ \\
\> \ \ \ \ \ and a set $A \in {\cal A}$ attaining the above minimum; \\
\> If $y_j < \delta$ then ~$\{ \delta \leftarrow y_j,
  \ A \leftarrow \{j\} \}$;\\
\> $y_i \leftarrow y_i + \delta, \ y_j \leftarrow y_j - \delta$; \\
\> Return $(y,A)$. \\
\> \\
{\em Algorithm} {\bf PipageRound}($(\cM,y)$): \\
\> While ($y$ is not integral) do \\
\> \> $T \leftarrow X$;  \\
\> \> While ($T$ contains fractional variables) do \\
\> \> \> Pick $i,j \in T$ fractional; \\
\> \> \> $(y^+, A^+) \leftarrow {\bf HitConstraint}(y,i,j)$; \\
\> \> \> $(y^-, A^-) \leftarrow {\bf HitConstraint}(y,j,i)$; \\
\> \> \> $p \leftarrow ||y^+ - y|| / ||y^+ - y^-||$; \\
\> \> \> With probability $p$, ~$\{ y \leftarrow y^-$,  $T \leftarrow T \cap A^- \}$; \\
\> \> \> ~~~~~~~~~~~~~~~~~~~Else ~~$\{ y \leftarrow y^+$, $T \leftarrow T \cap A^+ \}$;\\
\> \> EndWhile \\
\> EndWhile \\
\> Output $y$.
\end{tabbing}

Subsequently \cite{Vondrak09}, pipage rounding was extended to the case when
the starting point is in the matroid polytope $P(\cM)$, rather than $B(\cM)$.
This is not an issue in \cite{CCPV09}, but it is necessary for applications
with non-monotone submodular functions \cite{Vondrak09} or with additional
constraints, such as in this paper.

The following procedure takes care of the case when we start with a fractional
solution $x \in P(\cM)$. It adjusts the solution in a randomized way so that
the expectation of each variable is preserved, and the new fractional solution
is in the base polytope of a (possibly reduced) matroid.

\vspace{-10pt}
\begin{tabbing}
\ \ \ \= \ \ \= \ \ \= \ \ \= \ \ \= \ \ \  \\
{\em Algorithm} {\bf Adjust}($(\cM,x)$): \\
\> While ($x$ is not in $B(\cM)$) do \\
\> \> If (there is $i$ and $\delta>0$ such that $x + \delta \be_i \in P(\cM)$) do \\
\> \> \> Let $x_{max} = x_i + \max \{ \delta: x+\delta \be_i \in P(\cM) \}$; \\
\> \> \> Let $p = x_i / x_{max}$; \\
\> \> \> With probability $p$, ~$\{ x_i \leftarrow x_{max} \}$;\\
\> \> \> ~~~~~~~~~~~~~~~~~~~Else ~~$\{ x_i \leftarrow 0 \}$;\\
\> \> EndIf \\
\> \> If (there is $i$ such that $x_i = 0$) do \\
\> \> \> Delete $i$ from $\cM$ and remove the $i$-coordinate from $x$. \\
\> \> EndIf \\
\> EndWhile \\
\> Output $(\cM,x)$.
\end{tabbing}

To summarize, the complete procedure works as follows.
For a given $x \in P(\cM)$, we run $(\cM',y):=${\bf Adjust}$(\cM,x)$,
followed by {\bf PipageRound}($(\cM',y)$). The outcome is a base in the restricted
matroid where some elements have been deleted, i.e. an independent set
in the original matroid.

\section{Proofs and generalizations for randomized swap rounding}\label{sec:AppSwapRound}

In this section we proof that randomized swap rounding satisfies the conditions
of Lemma~\ref{lem:negCorr} and generalize the procedure to points in the matroid polytope.

\subsection{Proof of conditions for negative correlation}

\begin{lemma}
  \label{lem:swapNegCorr}
  Randomized swap rounding satisfies the conditions of Lemma~\ref{lem:negCorr}.
\end{lemma}
\begin{proof}
Let $X_{i,t}$ denote the $i$-th component of $\bX_t$. To prove the first
condition of Lemma~\ref{lem:negCorr} we condition on a
particular vector $\bX_t$ at time $t$ of the process and on its
convex representation $\bX_t=\sum_{\ell=1}^k \beta_\ell \b1_{B_\ell}$.  The
vector $\bX_{t+1}$ is obtained from $\bX_{t}$ by an elementary
operation. Without loss of generality we assume that the elementary
operation does a swap between the bases $B_1$ and $B_2$ involving
the elements $i\in B_1\setminus B_2$ and $j\in B_2\setminus B_1$. Let
$B_1'$ and $B_2'$ be the bases after the swap. Hence, with
probability $\beta_1/(\beta_1+\beta_2)$, $B_1'=B_1$ and
$B_2'=B_2-j+i$, and with probability $\beta_2/(\beta_1+\beta_2)$,
$B_1'=B_1-i+j$ and $B_2'=B_2$.  Thus,
\begin{align*}
\E[\beta_1 \b1_{B_1'} + \beta_2 \b1_{B_2'}]&=
\frac{\beta_1}{\beta_1+\beta_2} (\beta_1 \b1_{B_1}+\beta_2 (\b1_{B_2}-\be_j+\be_i))+
\frac{\beta_2}{\beta_1+\beta_2} (\beta_1 (\b1_{B_1}-\be_i+\be_j)+\beta_2 \b1_{B_2})\\
&=\beta_1\b1_{B_1} + \beta_2\b1_{B_2},
\end{align*}
where $\be_i=\b1_{\{i\}}$ and $\be_j=\b1_{\{j\}}$ denote the canonical basis vectors
corresponding to element $i$ and $j$, respectively.
Since the vector $\bX_{t+1}$ is given by $\bX_{t+1}=\beta_1 \b1_{B_1'} + \beta_2
\b1_{B_2'} + \sum_{\ell=3}^k \beta_\ell \b1_{B_\ell}$, we obtain
$\E[\bX_{t+1}\mid \bX_t]=\bX_t$.
The second condition of Lemma~\ref{lem:negCorr} is satisfied since an
elementary operation only changes two elements in one base of the convex
representation as discussed above.  To check the third condition of
the lemma, assume without loss of generality that $\bX_{t+1}$ is
obtained from $\bX_{t}=\sum_{\ell=1}^k \beta_\ell \b1_{B_\ell}$ by replacing
$B_1$ by $B_1-i+j$. Hence, $X_{i,t+1}=X_{i,t}+\beta_1$ and
$X_{j,t+1}=X_{j,t}-\beta_1$, implying that the third condition of the
lemma is satisfied.
\end{proof}

\subsection{Adapting randomized swap rounding to points in the matroid polytope}
\label{sec:swapMP}

In this section we show how randomized swap rounding can be generalized to
round a point in the matroid polytope to an independent set, such that
the conditions of Lemma~\ref{lem:negCorr} are still satisfied.  We
first present a generalization where the rounding is done by
applying randomized swap rounding for base polytopes to an
extension of the underlying matroid.
In a second step we show that this procedure can easily be interpreted
as a procedure on the initial matroid, leading to a simpler
description of the method.  An advantage of presenting the method as a
special case of base rounding, is that results presented for randomized
swap rounding on base polytopes easily carry over to the general
rounding procedure.

Let $x\in P(\cM)$ be the point to round. Similar as for the base
polytope case, we need a representation of $x$ as a convex combination
of independent sets.  Again, the algorithm of
Cunningham~\cite{Cunningham84} can be used to obtain a convex
combination of $x$ using at most $n+1$ independent sets with a running
time which is bounded by $O(n^6)$ oracle calls. Thus, we assume that
such a convex combination of $x$ using $n+1$ independent sets
$I_1,\dots, I_{n+1}\in \cI$ is given, i.e., $x=\sum_{\ell=1}^{n+1} \beta_\ell
\b1_{I_\ell}$.

Let $\cM'=(N',\cI')$ be the following extension of the matroid
$\cM=(N,\cI)$.  The set $N'$ is obtained from $N$ by adding $d$
additional dummy elements $\{s_1,\dots, s_d\}$, $N'=N\cup \{s_1,\dots,
s_d\}$. The independent sets are defined by $\cI'=\{I\subseteq N'\mid
I\cap N \in \cI, |I|\leq d\}$. Thus, a base of $\cM$ is also a base of
$\cM'$. The task of rounding $x$ in $\cM$ can be transformed into
rounding a point in the base polytope of $\cM'$ as follows. Every
independent set $I_\ell$ that is used in the convex representation of
$x$, is extended to a base $B_\ell'$ of $\cM'$ by adding an arbitrary
subset of $\{s_1,\dots,s_d\}$ of cardinality $d-|I_\ell|$. Hence,
$y=\sum_{\ell=1}^{n+1}\beta_\ell \b1_{B_\ell'}$ is a point in the base polytope
of $\cM'$. Then the randomized swap rounding procedure as presented in
Section~\ref{sec:swap-rounding} for points in the base polytope
is used to get a point $\b1_{B'}$ in $B(\cM')$. The point $\b1_{B'}$
is finally transformed into a point $\overline{x}$ that is a vertex of
$P(\cM)$ by projecting $\b1_{B'}$ onto the components corresponding to
elements in $N$. The point $\overline{x}$ is returned by the
algorithm.  By Lemma~\ref{lem:swapNegCorr}, the random point
$\b1_{B'}$ satisfies the conditions of Lemma~\ref{lem:negCorr}.  Since
the projection does not change the distribution of the components of
$\b1_{B'}$, also $\overline{x}$ satisfies the same properties.

The dummy elements can be interpreted as elements that do not have any
influence in the final outcome, since they will be removed by the
projection.  Consider for example an elementary operation on two bases
$B_1', B_2'\in \cB$ which are extensions of two independent set $I_1,
I_2 \in \cI$ to the matroid $\cM'$, and let $i\in B_1'\setminus B_2'$
and $j\in B_2'\setminus B_1'$ be the two elements involved in the
swap. If $i$ is a dummy element, i.e., $i\in \{s_1,\dots,s_d\}$, then
replacing $B_2'$ by $B_2'-j+i$ corresponds to removing element $j$
from $I_2$.

Consider the above algorithm using dummy elements with the following
modification: At each elementary operation, if possible, two non-dummy
elements are chosen.
One can easily observe that describing this version of the algorithm
without dummy elements corresponds to replacing the {\bf MergeBases}
procedure with the following procedure to merge two independent
sets. The procedure, called {\bf MergeIndepSets}, takes two
independent sets $I_1,I_2\in \cI$ and two positive scalars $\beta_1,
\beta_2$ as input. To simplify the description of the procedure, we
assume $|I_1|\geq|I_2|$, otherwise the roles of $I_1$ and $I_2$ have
to be exchanged in the algorithm.
\vspace{-10pt}
\begin{tabbing}
\ \ \ \= \ \ \= \ \ \= \ \ \= \ \ \= \ \ \  \\
{\em Algorithm} {\bf MergeIndepSets}$(\beta_1, I_1, \beta_2, I_2)$: \\
\> Find a set $S\subseteq I_1\setminus I_2$ of cardinality $|I_1|-|I_2|$ such that $I_2\cup S \in \cI$;\\
\> $I_2'=I_2\cup S$;\\
\> While ($I_1\neq I_2'$) do\\
\> \> Pick $i\in I_1\setminus I_2'$ and find $j\in I_2'\setminus I_1$ such that $I_1-i+j\in \cI$
and $I_2'-j+i\in \cI$;\\
\> \> With probability $\beta_1/(\beta_1+\beta_2)$, ~$\{I_2' \leftarrow I_2'-j+i\}$;\\
\> \> ~~~~~~~~~~~~~~~~~~~~~~~~~~~~~~~~~~~Else \hspace{0.08em}~$\{ I_1 \leftarrow I_1-i+j\}$;\\
\> EndWhile\\
\> For ($i \in S$) do\\
\> \> With probability $\beta_2/(\beta_1+\beta_2)$, ~$\{I_1 \leftarrow I_1-i\}$;\\
\> EndFor \\
\> Output $I_1$.
\end{tabbing}
The existence of a set $S$ as used in the algorithm easily follows from the
matroid axioms~\cite{Schrijver}. It can be found by successively
choosing elements in $I_1\setminus I_2$ that can be added to $I_2$
still maintaining independence.
Once the element $i\in I_1\setminus I_2'$ is chosen in the while loop
of the algorithm, the existence of an element $j\in I_2'\setminus I_1$
satisfying $I_1-i+j\in \cI$ and $I_2'-j+i\in \cI$ is guaranteed by
applying Theorem~\ref{thm:strongExchange} to the matroid
$\cM'=(N,\cI')$ given by $\cI'=\{I\in \cI \mid |I|\leq |I_1|\}$.

\section{Chernoff bounds for submodular functions}
\label{sec:submod-chernoff}

Here we prove Theorem~\ref{thm:submod-chernoff}, a Chernoff-type bound
for a monotone submodular function $f(X_1,\ldots,X_n)$ where $X_1,\ldots,X_n \in \{0,1\}$
are independent random variables.
Similarly to the proof of Chernoff bounds for linear functions,
the main trick is to prove a bound on the exponential moments
$\E[e^{\lambda f(X_1,\ldots,X_n)}]$. For that purpose, we write the value
of $f(X_1,\ldots,X_n)$ as follows: $f(X_1,\ldots,X_n) = \sum_{i=1}^{n} Y_i$,
where
$$Y_i = f(X_1,\ldots,X_i,0,\ldots,0)- f(X_1,\ldots,X_{i-1},0,\ldots,0).$$
The new complication is that the variables $Y_i$ are not independent.
There could be negative
and even positive correlations between $Y_i,Y_j$. What is important for us,
however, is that we can show negative correlation between $e^{\lambda \sum_{i=1}^{k-1} Y_i}$
and $e^{\lambda Y_k}$, and by induction the following bound.

\begin{lemma}
\label{lemma:Y-neg-correl}
For any $\lambda \in \RR$, a monotone submodular function and $Y_1,\ldots,Y_n$ defined as above,
$$ \E[e^{\lambda \sum_{i=1}^{n} Y_i}] \leq \prod_{i=1}^{n} \E[e^{\lambda Y_i}].$$
\end{lemma}

\begin{proof}
Denote $p_i = \Pr[X_i = 1]$. For any $k$, we have
\begin{eqnarray*}
\E[e^{\lambda \sum_{i=1}^{k} Y_i}]
 & = & \E[e^{\lambda f(X_1,\ldots,X_k,0,\ldots,0)}] \\
 & = &  p_k \ \E[e^{\lambda f(X_1,\ldots,X_{k-1},1,\ldots,0)}]
 + (1-p_k) \E[e^{\lambda f(X_1,\ldots,X_{k-1},0,\ldots,0)}] \\
 & = & p_k \ \E[e^{\lambda f(X_1,\ldots,X_{k-1},0,\ldots,0)}
 e^{\lambda F_k(X_1,\ldots,X_{k-1},0,\ldots,0)}]
 + (1-p_k) \E[e^{\lambda f(X_1,\ldots,X_{k-1},0,\ldots,0)}]
\end{eqnarray*}
where
$$ F_k(X_1,\ldots,X_{k-1},0,\ldots,0) =
f(X_1,\ldots,X_{k-1},1,\ldots,0)- f(X_1,\ldots,X_{k-1},0,\ldots,0) $$
denotes the marginal value of $X_k$ being set to $1$, given the
preceding variables. Observe that
$\E[F_k(X_1,\ldots,X_{k-1},$ $0,\ldots,0)]$ $ = \E[Y_k \mid X_k=1]$.

By submodularity, $F_k$ is a decreasing function of
$(X_1,\ldots,X_{k-1})$. On the other hand, $\sum_{i=1}^{k-1} Y_i =
 f(X_1,$ $\ldots,X_{k-1},0,\ldots,0)$ is an increasing function of
$(X_1,\ldots,X_{k-1})$.  We get the same monotonicity properties for
the exponential functions $e^{\lambda f(\ldots)}$ and $e^{\lambda
  F_k(\ldots)}$ (with a switch in monotonicity for $\lambda < 0$).  By
the FKG inequality, $e^{\lambda f(X_1,\ldots,X_{k-1},0,\ldots,0)}$
and $e^{\lambda F_k(X_1,\ldots,X_{k-1},0,\ldots,0)}$ are negatively
correlated, and we get
\begin{eqnarray*}
\E[e^{\lambda f(X_1,\ldots,X_{k-1},0,\ldots,0)}
   e^{\lambda F_k(X_1,\ldots,X_{k-1},0,\ldots,0)}] & \leq & \E[e^{\lambda f(X_1,\ldots,X_{k-1},0,\ldots,0)}]
   \E[e^{\lambda F_k(X_1,\ldots,X_{k-1},0,\ldots,0)}] \\
& = & \E[e^{\lambda \sum_{i=1}^{k-1} Y_i}] \E[e^{\lambda Y_k} \mid X_k=1].
\end{eqnarray*}
Hence, we have
\begin{eqnarray*}
\E[e^{\lambda \sum_{i=1}^{k} Y_i}] & \leq &
p_k \ \E[e^{\lambda \sum_{i=1}^{k-1} Y_i}]  \E[e^{\lambda Y_k} \mid X_k=1]
  + (1-p_k) \E[e^{\lambda \sum_{i=1}^{k-1} Y_i}] \\
& = &  \E[e^{\lambda \sum_{i=1}^{k-1} Y_i}] \cdot
(p_k \ \E[e^{\lambda Y_k} \mid X_k=1] + (1-p_k) \cdot 1) \\
& = &  \E[e^{\lambda \sum_{i=1}^{k-1} Y_i}] \cdot
(p_k \ \E[e^{\lambda Y_k} \mid X_k=1] + (1-p_k) \ \E[e^{\lambda Y_k} \mid X_k=0]) \\
& = &  \E[e^{\lambda \sum_{i=1}^{k-1} Y_i}] \cdot \E[e^{\lambda Y_k}].
\end{eqnarray*}
By induction, we obtain the lemma.
\end{proof}

Given this lemma, we can finish the proof of Theorem~\ref{thm:submod-chernoff}
following the same outline as of proof of the Chernoff bound.

\begin{proof}
  Let $Y_i = f(X_1,\ldots,X_k,0,\ldots,0)-
  f(X_1,\ldots,X_{k-1},0,\ldots,0)$ as above.  Let us denote $\E[Y_i]
  = \omega_i$ and $\mu = \sum_{i=1}^{n} \omega_i =
  \E[f(X_1,\ldots,X_n)]$.  By the convexity of the exponential and the
  fact that $Y_i \in [0,1]$,
$$ \E[e^{\lambda Y_i}] \leq \omega_i e^\lambda + (1-\omega_i)
 = 1 + (e^\lambda - 1) \omega_i \leq e^{(e^\lambda-1) \omega_i}.$$
Lemma~\ref{lemma:Y-neg-correl} then implies
$$ \E[e^{\lambda f(X_1,\ldots,X_n)}] = \E[e^{\lambda \sum_{i=1}^{n} Y_i}]
 \leq \prod_{i=1}^{n} \E[e^{\lambda Y_i}] \leq e^{(e^\lambda-1) \mu}.$$
For the upper-tail bound, we use Markov's inequality as follows:
$$ \Pr[f(X_1,\ldots,X_n) \geq (1+\delta) \mu]
 = \Pr[e^{\lambda f(X_1,\ldots,X_n)} \geq e^{\lambda (1+\delta) \mu}]
 \leq \frac{\E[e^{\lambda f(X_1,\ldots,X_n)}]}{e^{\lambda (1+\delta) \mu}}
 \leq \frac{e^{(e^\lambda-1) \mu}}{e^{\lambda (1+\delta) \mu}}.$$
We choose $e^\lambda = 1 + \delta$ which yields
$$ \Pr[f(X_1,\ldots,X_n) \geq (1+\delta) \mu]
 \leq \frac{e^{\delta \mu}}{(1+\delta)^{(1+\delta)\mu}}.$$
For the lower-tail bound, we use Markov's inequality with $\lambda < 0$ as follows:
$$ \Pr[f(X_1,\ldots,X_n) \leq (1-\delta) \mu]
 = \Pr[e^{\lambda f(X_1,\ldots,X_n)} \geq e^{\lambda (1-\delta) \mu}]
 \leq \frac{\E[e^{\lambda f(X_1,\ldots,X_n)}]}{e^{\lambda (1-\delta) \mu}}
 \leq \frac{e^{(e^{\lambda}-1) \mu}}{e^{\lambda (1-\delta) \mu}}.$$
We choose $e^{\lambda} = 1 - \delta$ which yields
$$ \Pr[f(X_1,\ldots,X_n) \leq (1-\delta) \mu]
 \leq \frac{e^{-\delta \mu}}{(1-\delta)^{(1-\delta)\mu}} \leq e^{-\mu \delta^2 / 2} $$
using $(1-\delta)^{1-\delta} \geq e^{-\delta + \delta^2/2}$
for $\delta \in (0,1]$.
\end{proof}

\section{Lower-tail estimate for submodular functions under dependent rounding}
\label{sec:submod-lower-tail}

\def\bb{{\bf b}}
\def\bc{{\bf c}}
\def\bx{{\bf x}}
\def\by{{\bf y}}

In this section, we prove Theorem~\ref{thm:swap-rounding-chernoff},
i.e. an exponential estimate for the lower tail
of the distribution of a monotone submodular function under randomized
swap rounding. We note that the bound on the expected value of the rounded solution,
$\E[f(R)] \geq \mu_0$, follows by the convexity of $F(x)$ along directions $\be_i - \be_j$
just like in \cite{CCPV09}; we omit the details. The exponential tail bound is much more
involved. We start by setting up some notation.

\paragraph{Notation.}
The rounding procedure starts from a convex linear combination of bases,
$$ \bx_0 = \sum_{i=1}^{n} \beta_i \b1_{B_i}.$$
The rounding proceeds in stages, where in the first stage we merge
the bases $B_1, B_2$ (randomly) into a new base $C_2$, and replace
$\beta_1 \b1_{B_1} + \beta_2 \b1_{B_2}$ in the linear combination
by $\gamma_2 \b1_{C_2}$, with $\gamma_2 = \beta_1 + \beta_2$.
More generally, in the $k$-th stage, we merge $C_k$ and $B_{k+1}$ into a new base
$C_{k+1}$ (we set $C_1=B_1$ in the first stage), and replace
$\gamma_k \b1_{C_k} + \beta_{k+1} \b1_{B_{k+1}}$
in the linear combination by $\gamma_{k+1} \b1_{C_{k+1}}$.
Inductively, $\gamma_{k+1} = \gamma_k + \beta_{k+1} = \sum_{i=1}^{k+1} \beta_i$.
After $n-1$ stages, we obtain a linear combination $\gamma_n \b1_{C_n}$
and $\gamma_n = \sum_{i=1}^{n} \beta_i = 1$; i.e., this is an integer solution.

We use the following notation to describe the vectors produced in the process:
\begin{itemize}
\item $\bb_i = \beta_i \b1_{B_i}$
\item $\bc_i = \gamma_i \b1_{C_i}$
\item $\by_k = \sum_{i=k}^{n} \bb_i = \sum_{i=k}^{n} \beta_i \b1_{B_i}$
\item $\bx_k = \bc_{k+1} + \by_{k+2} =
\gamma_{k+1} \b1_{C_{k+1}} +\sum_{i=k+2}^{n} \beta_i \b1_{B_i}$
\end{itemize}
In other words, $\bb_i$ are the initial vectors in the linear combination,
which get gradually replaced by $\bc_i$, and $\bx_k$ is the fractional
solution after $k$ stages.

We emphasize that $\bx_k$ denotes the entire fractional solution at
a certain stage and not the value of its $k$-th coordinate.
The coordinates of the fractional solution are the variables $X_i$.
If we want to refer to the value of $X_i$ after $k$ stages,
we use the notation $X_{i,k}$.

We work with the multilinear extension of a submodular function,
$F(\bx) = \E[f(\hat{\bx})]$. In the following, we use the following
shorthand notation and basic properties:
\begin{itemize}
\item $F_i(\bx)$ denotes the partial derivative $\partdiff{F}{X_i}$
evaluated at $\bx$. The interpretation of $F_i(\bx)$ is
the marginal value of $i$ with respect to the fractional solution $\bx$.
\item We use $\be_i = \b1_{\{i\}}$ to denote the canonical basis vector
corresponding to element $i$.
\item If only one variable is changing while others are fixed,
$F(\bx)$ is a linear function. Therefore, we can use the following
formula:
$$ F(\bx + t \be_i) = F(\bx) + t F_i(\bx).$$
\item Due to submodularity, $\mixdiff{F}{X_i}{X_j} \leq 0$ for any $i,j$.
This implies that $F_i(\bx) = \partdiff{F}{X_i}$ is non-increasing
as a function of each coordinate of $\bx$. If $\by$ dominates $\bx$
in all coordinates ($\bx \leq \by$), we have $F_i(\bx) \geq F_i(\by)$.
\end{itemize}

\paragraph{Proof overview.}
The random process in terms of the evolution of $F(x)$ is a submartingale,
i.e. the value in each step can only increase in expectation.
This is a good sign;
however, a straightforward application of concentration bounds for
martingales yields a dependency of the number of variables $n$ which
would render the bound meaningless.
More refined bounds for martingales rely on bounds
on the variance in successive steps. Unfortunately,
these are also difficult to use since we do not have a good a priori
bound on the variance in each step.
The variance can depend on preceding steps and taking worst-case bounds
leads to the same dependency on $n$ as mentioned above.

In order to prove a bound which depends only on the parameters $\delta$
and $\mu_0$, we start from scratch and follow the standard recipe:
estimate the exponential moment $\E[e^{\lambda (\mu_0 - f(R))}]$,
where $\mu_0$ is the initial value and $R$ is the rounded solution.
We decompose the expression $e^{\lambda (\mu_0 - f(R))}$ into
a telescoping product:
$$ e^{\lambda (\mu_0 - f(R))} = e^{\lambda (F(\bx_0) - F(\bx_{n-1}))}
 = e^{\lambda (F(\bx_0) - F(\bx_1))} \cdot e^{\lambda (F(\bx_1) - F(\bx_2))}
 \cdot \ldots \cdot e^{\lambda (F(\bx_{n-2}) - F(\bx_{n-1}))}.$$
The factors in this product are not independent, but we can
prove bounds on the conditional expectations
$\E[e^{\lambda (F(\bx_{k-1}) - F(\bx_k))} \mid \bx_0,\ldots,\bx_{k-1}]$,
in other words conditioned on a given history of the rounding process.
These bounds depend on the history, but we are able to charge
the arising factors to the value of $\mu_0 = F(x_0)$ in such a way
that the final bound depends only on $\mu_0$.

We start from the bottom, by analyzing the basic rounding step for
two variables. The following elementary inequality will be helpful.

\begin{lemma}
\label{lemma:exp-bound}
For any $p \in [0,1]$ and $\xi \in [-1,1]$,
$$ p e^{\xi (1-p)} + (1-p) e^{-\xi p} \leq e^{\xi^2 p(1-p)}.$$
\end{lemma}

\begin{proof}
If $\xi < 0$, we can replace $\xi$ by $-\xi$ and $p$ by $1-p$;
the statement of the lemma remains the same. So we can assume $\xi \in [0,1]$.

Fix any $p \in [0,1]$ and
define $\phi_p(\xi) = e^{\xi^2 p(1-p)} - p e^{\xi(1-p)} -
 (1-p) e^{-\xi p}$. It is easy to see that $\phi_p(0) = 0$.
Our goal is to prove that $\phi_p(\xi) \geq 0$ for $\xi \in [0,1]$.
Let us compute the derivative of $\phi_p(\xi)$ with respect to $\xi$:
\begin{eqnarray*}
\phi'_p(\xi) & = & 2 \xi p(1-p) e^{\xi^2 p(1-p)}
 - p(1-p) e^{\xi(1-p)} + p(1-p) e^{-\xi p} \\
 & = & p(1-p) e^{-\xi p} \left( 2 \xi e^{\xi^2 p(1-p) + \xi p}
 - e^{\xi} + 1 \right) \\
 & \geq & p(1-p) e^{-\xi p} \left( 2 \xi - e^\xi + 1 \right).
\end{eqnarray*}
For $\xi \in [0,1]$, we have $e^\xi \leq 1 + 2\xi$ and hence
$\phi'_p(\xi) \geq 0$. This means that $\phi_p(\xi)$ is non-decreasing
and $\phi_p(\xi) \geq 0$ for $\xi \in [0,1]$.
\end{proof}

Note that the lemma does not hold for arbitrarily large $\xi$,
e.g. when $p = 1/\xi^2$ and $\xi \rightarrow \infty$.
Next, we apply this lemma to the basic step of the rounding
procedure.

\begin{lemma}
\label{lemma:one-step}
Let $F(\bx)$ be the multilinear extension of a monotone submodular function with marginal
values in $[0,1]$, and let $\lambda \in [0,1]$.
Consider one elementary operation of randomized swap rounding,
where two variables $X_i,X_j$ are modified.
Let $\bx$ denote the fractional solution before, $\bx'$ after this step,
and let $\cal H$ denote the complete history prior to this rounding step.
Assume that the values of the two variables before the rounding step
are $X_i = \gamma, X_j = \beta$. Then
$$ \E[e^{\lambda(F(\bx) - F(\bx'))} \mid {\cal H}] \leq
 e^{\lambda^2 \beta \gamma (F_j(\bx) - F_i(\bx))^2} $$
where $F_i(\bx) = \partdiff{F}{X_i}(\bx)$ and $F_j(\bx) = \partdiff{F}{X_j}(\bx)$.
\end{lemma}

\begin{proof}
Fix the history $\cal H$; this includes the point $\bx$ before the rounding step.
With probability $p = \frac{\gamma}{\beta+\gamma}$,
the rounding step is $X'_i = X_i + \beta$ and $X'_j = X_j - \beta$.
I.e., $\bx' = \bx + \beta \be_i - \beta \be_j$. Since $F(\bx)$
is linear when only one coordinate is modified, we get
$$ F(\bx') = F(\bx) + \beta F_i(\bx) - \beta F_j(\bx+\beta \be_i).$$
By submodularity, $F_j(\bx+\beta \be_i) \leq F_j(\bx)$ and hence
$$ F(\bx') = F(\bx) + \beta F_i(\bx) - \beta F_j(\bx+\beta \be_i)
 \geq F(\bx) + \beta (F_i(\bx) - F_j(\bx)).$$
With probability $1-p$, we set $X'_i = X_i - \gamma$ and $X'_j = X_j + \gamma$.
By similar reasoning, in this case we get
$$ F(\bx') = F(\bx) - \gamma F_i(\bx) + \gamma F_j(\bx - \gamma \be_i)
\geq F(\bx) - \gamma (F_i(\bx) - F_j(\bx)).$$
Taking expectation over the two cases, we get
\begin{eqnarray*}
\E[e^{\lambda(F(\bx) - F(\bx'))} \mid {\cal H}]
 & \leq & p e^{\lambda \beta (F_j(\bx) - F_i(\bx))}
 + (1-p) e^{-\lambda \gamma (F_j(\bx) - F_i(\bx))} \\
 & = & p e^{\lambda (1-p)(\beta+\gamma) (F_j(\bx) - F_i(\bx))}
 + (1-p) e^{-\lambda p (\beta+\gamma) (F_j(\bx) - F_i(\bx))}.
\end{eqnarray*}
We invoke Lemma~\ref{lemma:exp-bound} with $\xi = \lambda (\beta+\gamma)
(F_j(\bx) - F_i(\bx))$ (we have $|\xi| \leq 1$ due to
$\lambda, \beta+\gamma, F_i(\bx), F_j(\bx)$ all being in $[0,1]$). We get
$$ \E[e^{\lambda(F(\bx) - F(\bx'))} \mid {\cal H}] \leq e^{\xi^2 p(1-p)}
 = e^{\lambda^2 \beta \gamma (F_j(\bx) - F_i(\bx))^2}. $$
\end{proof}

Note that the exponent on the right-hand side of Lemma~\ref{lemma:one-step}
corresponds to the variance in one step of the rounding procedure.
The next lemma estimates these contributions,
aggregated over one stage of the rounding process,
i.e., the merging of the bases $C_k$ and $B_{k+1}$.
The exponent on the right-hand side of Lemma~\ref{lemma:stage-bound}
corresponds to the variance of the random process accumulated over the $k$-th stage.
It is crucial that we compare this quantity to certain values
which can be eventually charged to $\mu_0$.

\begin{lemma}
\label{lemma:stage-bound}
Let $F(\bx)$ be the multilinear extension of a monotone submodular function with marginal
values in $[0,1]$, and let $\lambda \in [0,1]$.
Consider the $k$-th stage of the rounding process, when bases
$C_k$ and $B_{k+1}$ (with coefficients $\gamma_k$ and $\beta_{k+1}$)
are merged into $C_{k+1}$. The fractional solution
before this stage is $\bx_{k-1}$ and after this stage $\bx_k$.
Conditioned on any history ${\cal H}$ of the rounding process
throughout the first $k-1$ stages,
$$ \E[e^{\lambda (F(\bx_{k-1}) - F(\bx_k))} \mid {\cal H}]
 \leq e^{\lambda^2 (\beta_{k+1} F(\bc_k) + \gamma_k(F(\by_{k+1}) - F(\by_{k+2})))}.$$
\end{lemma}

\begin{proof}
The $k$-th stage merges bases $C_k$ and $B_{k+1}$ into $C_{k+1}$
by taking elements in pairs and performing rounding steps as in
Lemma~\ref{lemma:one-step}. Let us denote the pairs of elements
considered by the rounding procedure $(c_1,b_1), \ldots, (c_d,b_d)$,
where $C_k = \{c_1,\ldots,c_d\}$ and $B_{k+1} = \{b_1,\ldots,b_d\}$.
The matching is not determined beforehand: $(c_2,b_2)$ might depend
on the random choice between $c_1,b_1$, etc.
In the following, we drop the index $k$ and denote by
$\bx^i$ the fractional solution obtained after processing $(c_1,b_1), \ldots, (c_i,b_i)$.
We start with $\bx^0 = \bx_{k-1}$ and after processing all $d$ pairs,
we get $\bx^d = \bx_k$. We also replace $\beta_{k+1}, \gamma_k$ simply by $\beta, \gamma$.
We denote by ${\cal H}_{i}$ the complete history prior to the rounding step
involving $(c_{i+1},b_{i+1})$; in particular, this includes
the fractional solution $\bx^i$.

Using Lemma~\ref{lemma:one-step} for the rounding step involving
$(c_{i+1},b_{i+1})$, we get
$$ \E[e^{\lambda(F(\bx^i) - F(\bx^{i+1}))} \mid {\cal H}_{i}]
\leq e^{\lambda^2 \gamma \beta (F_{c_{i+1}}(\bx^{i}) -
      F_{b_{i+1}}(\bx^{i}))^2}
\leq e^{\lambda^2 \gamma \beta (F_{c_{i+1}}(\bx^{i}) +
      F_{b_{i+1}}(\bx^{i}))}, $$
using the fact that the partial derivatives $F_j(\bx^{i})$ are in $[0,1]$.

Further, we modify the exponent of the right-hand side as follows.
The vector $\bx^{i}$ is obtained after processing
$i$ pairs and still contains the coordinates $c_{i+1},\ldots,c_d$
of $\bc_k = \gamma \b1_{C_k}$ untouched: in other words,
$\bx^{i} \geq \gamma \b1_{\{c_{i+1},\ldots,c_d\}}$.
Let us define
\begin{itemize}
\item $\bc^i 
 = \gamma \b1_{\{ c_{i+1},\ldots,c_d\}}$.
\end{itemize}
I.e., $\bx^{i} \geq \bc^i \geq \bc^{i+1}$. By submodularity, we have
$F_{c_{i+1}}(\bx^{i}) \leq F_{c_{i+1}}(\bc^{i+1})$.

Similarly, the vector $\bx^{i}$ also contains
the coordinates $b_{i+1},\ldots,b_d$ of $\bb_{k+1}$ and all of
$\by_{k+2} = \sum_{j=k+2}^{n} \bb_j$ unchanged:
$\bx^{i} \geq \beta \b1_{\{b_{i+1},\ldots,b_d\}} + \by_{k+2}$.
Let us define
\begin{itemize}
\item $\by^i
 = \beta \b1_{\{b_{i+1},\ldots,b_d\}} + \by_{k+2}$.
\end{itemize}
I.e., $\bx^{i} \geq \by^{i} \geq \by^{i+1}$. By submodularity, we get
$F_{b_{i+1}}(\bx^{i}) \leq F_{b_{i+1}}(\by^{i+1})$.
Therefore, we can write
\begin{equation}
\label{eq:swap-exp-bound}
\E[e^{\lambda(F(\bx^{i}) - F(\bx^{i+1}))} \mid {\cal H}_{i}]
\leq e^{\lambda^2 \gamma \beta
(F_{c_{i+1}}(\bc^{i+1}) + F_{b_{i+1}}(\by^{i+1}))}.
\end{equation}
We claim that by induction on $d-i$, this implies
\begin{equation}
\label{eq:induct-bound}
\E[e^{\lambda(F(\bx^{i}) - F(\bx^d))} \mid {\cal H}_{i}]
\leq e^{\lambda^2 (\beta F(\bc^{i}) + \gamma(F(\by^{i}) - F(\by^{d})))}
\end{equation}
for all $i=0,\ldots,d$.
For $i=d$, the claim is trivial.
For $i<d$, we can write
\begin{eqnarray*}
\E[e^{\lambda(F(\bx^{i}) - F(\bx^d))} \mid {\cal H}_{i}]
 & = & \E\left[e^{\lambda(F(\bx^{i}) - F(\bx^{i+1}))}
  \E[e^{\lambda(F(\bx^{i+1}) - F(\bx^d))} \mid {\cal H}_{i+1}]
 \ \big| \ {\cal H}_{i}\right]
\end{eqnarray*}
and using the inductive hypothesis (\ref{eq:induct-bound}) for $i+1$,
\begin{eqnarray*}
 \E[e^{\lambda(F(\bx^{i}) - F(\bx^d))} \mid {\cal H}_{i}]
 & \leq & \E\left[e^{\lambda(F(\bx^{i}) - F(\bx^{i+1}))} \cdot
e^{\lambda^2 (\beta F(\bc^{i+1}) + \gamma(F(\by^{i+1}) - F(\by^{d})))}
 \mid {\cal H}_{i}\right] \\
& = & e^{\lambda^2 (\beta F(\bc^{i+1}) +
 \gamma(F(\by^{i+1}) - F(\by^{d})))}
 \cdot \E\left[e^{\lambda(F(\bx^{i}) - F(\bx^{i+1}))} \mid
 {\cal H}_{i}\right]
\end{eqnarray*}
where we used the fact that the inductive bound is determined by
${\cal H}_{i}$, and so we can take it out of the expectation
(it depends only on the sets $\{c_{i+2},\ldots,c_d\}$
and $\{b_{i+2}, \ldots, b_d\}$ which are determined even before
performing the rounding step on $(c_{i+1},b_{i+1})$). Taking logs
and using (\ref{eq:swap-exp-bound}) to estimate the last expectation,
we obtain
\begin{eqnarray*}
& & \log \E[e^{\lambda(F(\bx^{i}) - F(\bx^{d}))} \mid {\cal H}_{i}] \\
& \leq & {\lambda^2 \left(\beta F(\bc^{i+1}) +
 \gamma\big(F(\by^{i+1}) - F(\by^{d})\big)\right)} +
{\lambda^2 \gamma \beta \Big(F_{c_{i+1}}(\bc^{i+1}) + F_{b_{i+1}}(\by^{i+1})\Big)} \\
& = & {\lambda^2 \left(\beta \left(F(\bc^{i+1}) + \gamma F_{c_{i+1}}(\bc^{i+1})\right)
+ \gamma\big(F(\by^{i+1}) + \beta F_{b_{i+1}}(\by^{i+1})
 - F(\by^{d})\big)\right)} \\
& = & {\lambda^2 \left(\beta F(\bc^{i}) + \gamma\big(F(\by^{i}) - F(\by^{d})\big)\right)}
\end{eqnarray*}
where we used $F(\bc^{i+1}) + \gamma F_{c_{i+1}}(\bc^{i+1}))
 = F(\bc^{i})$ and $F(\by^{i+1}) + \beta F_{b_{i+1}}(\by^{i+1})
 = F(\by^{i})$ (see the definitions of $\bc^i, \by^i$ above).

This proves our inductive claim (\ref{eq:induct-bound}).
For $i=0$,
since $\bx^0 = \bx_{k-1}$, $\bx^d = \bx_k$, $\bc^0 = \bc_k$, $\by^0 = \by_{k+1}$ and $\by^{d} = \by_{k+2}$,
this gives the statement of the lemma.
\end{proof}

Now we can proceed finally to the proof of
Theorem~\ref{thm:swap-rounding-chernoff}.

\begin{proof}
We prove inductively the following statement: For any $k$
and any $\lambda \in [0,1]$,
\begin{equation}
\label{eq:stage-induct}
\E[e^{\lambda (\mu_0 - F(\bx_k))}]
 \leq e^{\lambda^2 (\mu_0 (1+\sum_{i=1}^{k} \beta_{i+1}) - F(\by_{k+2}))}.
\end{equation}
We remind the reader that $\mu_0 = F(\bx_0)$, $\bx_k$ is the fractional solution
after $k$ stages, and $\by_{k+2} = \sum_{i=k+2}^{n} \bb_i$.
We proceed by induction on $k$.

For $k=0$, the claim is trivial, since $F(\by_2) \leq F(\bx_0) = \mu_0$
by monotonicity. For $k \geq 1$, we unroll the expectation as follows:
$$ \E[e^{\lambda (\mu_0 - F(\bx_k))}] =
\E\left[e^{\lambda (\mu_0 - F(\bx_{k-1}))}
    \E[e^{\lambda (F(\bx_{k-1}) - F(\bx_k))} \mid {\cal H}]
\right] $$
where $\cal H$ is the complete history prior to stage $k$ (up to $\bx_{k-1}$).
We estimate the inside expectation using Lemma~\ref{lemma:stage-bound}:
$$ \E[e^{\lambda (F(\bx_{k-1}) - F(\bx_k))} \mid {\cal H}]
 \leq e^{\lambda^2 (\beta_{k+1} F(\bc_k) + \gamma_k(F(\by_{k+1}) - F(\by_{k+2})))}
 \leq e^{\lambda^2 (\beta_{k+1} F(\bx_{k-1}) + F(\by_{k+1}) - F(\by_{k+2}))}
$$
using monotonicity, $\bc_k \leq \bx_{k-1}$, $\by_{k+2} \leq \by_{k+1}$ and $\gamma_k \leq 1$.
Therefore,
\begin{eqnarray*}
\E[e^{\lambda (\mu_0 - F(\bx_k))}] & \leq &
\E\left[e^{\lambda (\mu_0 - F(\bx_{k-1}))}
 e^{\lambda^2 (\beta_{k+1} F(\bx_{k-1}) + F(\by_{k+1}) - F(\by_{k+2}))}
\right] \\
& = &
 e^{\lambda^2 (\beta_{k+1} \mu_0 + F(\by_{k+1}) - F(\by_{k+2}))}
\E\left[e^{(\lambda - \lambda^2 \beta_{k+1})(\mu_0 - F(\bx_{k-1}))}
\right].
\end{eqnarray*}
By the inductive hypothesis (\ref{eq:stage-induct})
with $\lambda' = \lambda - \lambda^2 \beta_{k+1} \in [0,1]$,
$$\E[e^{(\lambda - \lambda^2 \beta_{k+1})(\mu_0 - F(\bx_{k-1}))}]
 \leq e^{(\lambda-\lambda^2 \beta_{k+1})^2
 (\mu_0 (1+\sum_{i=1}^{k-1} \beta_{i+1}) - F(\by_{k+1}))}
\leq e^{\lambda^2 (\mu_0 (1+\sum_{i=1}^{k-1} \beta_{i+1}) - F(\by_{k+1}))}.$$
In the last inequality we used $F(\by_{k+1}) \leq \mu_0$, which holds by monotonicity.
Plugging this into the preceding equation,
\begin{eqnarray*}
\E[e^{\lambda (\mu_0 - F(\bx_k))}] & \leq &
 e^{\lambda^2 (\beta_{k+1} \mu_0 + F(\by_{k+1}) - F(\by_{k+2}))}
e^{\lambda^2 (\mu_0 (1+\sum_{i=1}^{k-1} \beta_{i+1}) - F(\by_{k+1}))} \\
& = & e^{\lambda^2 (\mu_0 (1+\sum_{i=1}^{k} \beta_{i+1}) - F(\by_{k+2}))}
\end{eqnarray*}
which proves (\ref{eq:stage-induct}). Finally, for $k=n-1$ we obtain
$F(\bx_{n-1}) = f(R)$ where $R$ is the rounded solution, $\by_{n+1} = 0$, and
\begin{equation}
\label{eq:exp-moment-bound}
\E[e^{\lambda (\mu_0 - f(R))}] \leq e^{\lambda^2 \mu_0 (1+\sum_{i=1}^{n-1} \beta_{i+1})}
 \leq e^{2 \lambda^2 \mu_0}
\end{equation}
because $\sum_{i=1}^{n-1} \beta_{i+1} \leq 1$.
The final step is to apply Markov's inequality to the exponential moment.
From Markov's inequality and Equation~(\ref{eq:exp-moment-bound}), we get
$$ \Pr[ f(R) \leq (1-\delta) \mu_0]
  = \Pr\left[e^{\lambda (\mu_0 - f(R))} \geq e^{\lambda \delta \mu_0}\right]
 \leq \frac{\E[e^{\lambda (\mu_0 - f(R))}]}{e^{\lambda \delta \mu_0}}
 \leq e^{2 \lambda^2 \mu_0 - \lambda \delta \mu_0}. $$
A choice of $\lambda = \delta / 4$ gives the statement of the theorem.
\end{proof}

\end{document}